\documentclass[aps,pra,twocolumn,nofootinbib, superscriptaddress,10pt]{revtex4-2}

\usepackage{graphicx}
\usepackage{epstopdf}
\usepackage{amsmath}
\usepackage{amssymb}
\usepackage{mathrsfs}
\usepackage{amsthm}
\usepackage{bm}
\usepackage{url}
\usepackage[T1]{fontenc}
\usepackage{csquotes}
\MakeOuterQuote{"}

\usepackage{color}

\newtheoremstyle{note}
  {\topsep/2}               
  {\topsep/2}               
  {}                      
  {\parindent}            
  {\itshape}              
  {.}                     
  {5pt plus 1pt minus 1pt}
  {}

\theoremstyle{note}
\newtheorem{theorem}{Theorem}
\newtheorem{lemma}{Lemma}

\newtheorem{corollary}{Corollary}
\newtheorem{proposition}{Proposition}

\theoremstyle{definition}

\theoremstyle{remark}

\newcommand{\half}{\mbox{$\textstyle \frac{1}{2}$}}

\newcommand{\spa}{\operatorname{span}}
\newcommand{\tr}{\operatorname{tr}}

\newcommand{\imply}{\mathrel{\Rightarrow}}

 \newcommand{\rme}{\mathrm{e}}
 \newcommand{\rmi}{\mathrm{i}}

 \newcommand{\rmA}{\mathrm{A}}
 \newcommand{\rmB}{\mathrm{B}}

 \newcommand{\rmE}{\mathrm{E}}

 \newcommand{\caH}{\mathcal{H}}

 \newcommand{\caS}{\mathcal{S}}
 \newcommand{\caU}{\mathcal{U}}

  \newcommand{\scrD}{\mathscr{D}}
  
 \newcommand{\scrS}{\mathscr{S}}
  \newcommand{\scrT}{\mathscr{T}}

 \newcommand{\eff}{\mathrm{eff}}

\newcommand{\be}{\begin{equation}}
\newcommand{\ee}{\end{equation}}
\newcommand{\ba}{\begin{align}}
\newcommand{\ea}{\end{align}}

\def\<{\langle}  
\def\>{\rangle}  

\def\eqref#1{\textup{(\ref{#1})}}  
\newcommand{\eref}[1]{Eq.~\textup{(\ref{#1})}}
\newcommand{\Eref}[1]{Equation~\textup{(\ref{#1})}}
\newcommand{\esref}[1]{Eqs.~\textup{(\ref{#1})}}
\newcommand{\Esref}[1]{Equations~\textup{(\ref{#1})}}

\newcommand{\fref}[1]{Fig.~\ref{#1}}
\newcommand{\Fref}[1]{Figure~\ref{#1}}

\newcommand{\sref}[1]{Sec.~\ref{#1}}
\newcommand{\Sref}[1]{Section~\ref{#1}}

\newcommand{\thref}[1]{Theorem~\ref{#1}}
\newcommand{\Thref}[1]{Theorem~\ref{#1}}

\newcommand{\lref}[1]{Lemma~\ref{#1}}
\newcommand{\Lref}[1]{Lemma~\ref{#1}}
\newcommand{\lsref}[1]{Lemmas~\ref{#1}}

\newcommand{\pref}[1]{Proposition~\ref{#1}}

\newcommand{\crref}[1]{Corollary~\ref{#1}}
\newcommand{\Crref}[1]{Corollary~\ref{#1}}

\newcommand{\cref}[1]{Conjecture~\ref{#1}}
\newcommand{\Cref}[1]{Conjecture~\ref{#1}}

\newcommand{\aref}[1]{Appendix~\ref{#1}}

\newcommand{\rcite}[1]{Ref.~\cite{#1}}
\newcommand{\rscite}[1]{Refs.~\cite{#1}}

\begin{document}
\title{Minimum number of experimental settings required to verify bipartite pure states and unitaries}

\author{Yunting Li}
\affiliation{State Key Laboratory of Surface Physics and Department of Physics, Fudan University, Shanghai 200433, China}
\affiliation{Institute for Nanoelectronic Devices and Quantum Computing, Fudan University, Shanghai 200433, China}
\affiliation{Center for Field Theory and Particle Physics, Fudan University, Shanghai 200433, China}

\author{Haoyu Zhang}
\affiliation{State Key Laboratory of Surface Physics and Department of Physics, Fudan University, Shanghai 200433, China}
\affiliation{Institute for Nanoelectronic Devices and Quantum Computing, Fudan University, Shanghai 200433, China}
\affiliation{Center for Field Theory and Particle Physics, Fudan University, Shanghai 200433, China}

\author{Zihao Li}
\affiliation{State Key Laboratory of Surface Physics and Department of Physics, Fudan University, Shanghai 200433, China}
\affiliation{Institute for Nanoelectronic Devices and Quantum Computing, Fudan University, Shanghai 200433, China}
\affiliation{Center for Field Theory and Particle Physics, Fudan University, Shanghai 200433, China}

\author{Huangjun Zhu}
\email{zhuhuangjun@fudan.edu.cn}
\affiliation{State Key Laboratory of Surface Physics and Department of Physics, Fudan University, Shanghai 200433, China}
\affiliation{Institute for Nanoelectronic Devices and Quantum Computing, Fudan University, Shanghai 200433, China}
\affiliation{Center for Field Theory and Particle Physics, Fudan University, Shanghai 200433, China}

\date{\today}

\begin{abstract}
Efficient verification of quantum states and gates is crucial to the development of quantum technologies. 	Although the sample complexities of quantum state verification and quantum gate verification have been studied by many researchers, the number of 	experimental settings has received little attention and is poorly understood. In this work we study systematically quantum state verification and quantum gate verification with a focus on the number of experimental settings. 
We show that any bipartite pure state can be verified by only two measurement settings based on local projective measurements. Any bipartite unitary in dimension $d$ can be verified by $2d$ experimental settings based on local operations.  In addition, we introduce the concept of entanglement-free verification and clarify its connection with minimal-setting verification. 
Finally, we show that any two-qubit unitary can be verified with at most five experimental settings; moreover, a generic two-qubit unitary (except for a set of measure zero)
can be verified by an entanglement-free protocol based on four settings.	In the course of study we clarify the properties of Schmidt coefficients of two-qubit unitaries, which are of independent interest.
\end{abstract}

\maketitle

\section{Introduction}
Quantum information processing has attracted increasing attention recently due to  its great potential and profound implications. 
To harness the power of quantum information processing, it is crucial to verify the underlying quantum states and devices efficiently based on the accessible measurements. Unfortunately, traditional tomographic approaches are notoriously inefficient since the resource overhead increases exponentially with the system size under consideration. To overcome this problem, a number of alternative approaches have been proposed recently; see \rscite{Eisert,KlieschTheory,Theo2021,yu2021statistical} for an overview.

Among alternative approaches proposed so far, \emph{quantum state verification} (QSV) is particularly  appealing because it can achieve a high efficiency based on
local operations and classical communication (LOCC) \cite{Masahito2005,Aolita2015,Many-qubit,Pallister2018Optimal,Zhu2019Efficient,Zhu-general}. 
Notably, efficient verification protocols based on local projective measurements have been constructed for bipartite pure states \cite{Masahito2005,Zhu-entangled,LiHZ19,MHtwo-qubit,S.JWbipartite}, stabilizer states \cite{Stabilizer2015,Pallister2018Optimal,Kalev2019,Zhu-hypergraph,Zhu-general,LiGHZ,Tom-stabilizer}, hypergraph states \cite{Zhu-hypergraph}, weighted graph states \cite{Hayashi_weighted}, and Dicke states \cite{Liu-dicke,LiHSS21}. Moreover, the efficiency of QSV has been demonstrated in a number of experiments \cite{ExpEntangled,ExpLu,TowardsJiang,ClassicalZhang}. 
Recently, the idea of QSV was generalized to  \emph{quantum gate verification} (QGV) \cite{Liu2020Efficient,ZhuZ20,Zeng2019Quantum} (cf.  \rscite{Hofmann2005Complementary,Daniel2013Minimum,Mayer2018Quantum,Wu_2019,Cross2020}), which enables efficient verification of various quantum gates and quantum circuits based on LOCC.
Notably, all bipartite unitaries and Clifford unitaries can be verified with resources that are independent of the system size, while the resource required to verify the generalized controlled-NOT (CNOT) gate and generalized controlled-$Z$ (CZ) gate grows only linearly with the system size. The efficiency of QGV has also been demonstrated in several experiments recently \cite{zhang2021efficient,luo2021proofofprinciple}.

So far most works on QSV and QGV have exclusively focused on the sample efficiency as the main figure of merit. By contrast, the number of experimental settings has received little attention, although this figure of merit is also of key interest to both theoretical study and practical applications.  Even for bipartite pure states,  it is still not clear how many measurement settings are required to construct a reliable verification protocol. The situation is even worse in the case of bipartite unitaries, not to mention the multipartite scenario. This problem becomes particularly important when it is difficult or slow to switch measurement settings, which is the case in many practical scenarios.

In this work we study systematically QSV and QGV with a focus on the number of experimental settings based on LOCC. We show that any bipartite pure state can be verified by two measurement settings based on nonadaptive local projective measurements. By contrast, at least $d$ experimental settings based on local operations  are required to verify each bipartite unitary in dimension $d$, while $2d$ settings are sufficient.  In addition, we introduce the concept of entanglement-free verification, which is of special interest to both theoretical study and practical applications. Moreover, we  show that any entanglement-free verification protocol can be turned into a minimal-setting protocol, and vice versa.

For each two-qubit unitary, we determine the minimum number of required experimental settings explicitly. Our study shows that any two-qubit unitary can be verified using only five experimental settings, while
a generic two-qubit unitary (except for a set of measure zero) can be verified by an  entanglement-free protocol based on four  settings. Explicit entanglement-free protocols are constructed  for CNOT, CZ,  controlled-phase (C-Phase), and SWAP gates, respectively. In the course of study we clarify the properties of Schmidt coefficients of two-qubit unitaries and their implications for studying the equivalence relation under local unitary transformations, which are of interest beyond the main focus of this work.

The rest of this paper is organized as follows. In \sref{sec:QSGV},  we briefly review the basic frameworks of QSV  and QGV. In \sref{sec:BipartitePureState}, we determine the minimum number of measurement settings required to verify each bipartite pure state. In \sref{sec:UMinS}, we clarify the relation between minimal-setting verification and entanglement-free verification; in addition, we derive nearly tight lower and upper bounds for the minimum number of settings required to verify each bipartite unitary. In \sref{sec:TwoQubitU}, we clarify the properties of  Schmidt coefficients of two-qubit unitaries. In \sref{sec:VTwoQuibtU}, we determine the minimum number of  settings required to verify each two-qubit unitary. \Sref{sec:summary} summarizes the paper. To streamline the presentation, some technical proofs are relegated to the appendixes.

\section{\label{sec:QSGV}Quantum state and gate  verification}
In preparation for the later study, here we briefly review the basic frameworks of QSV \cite{Pallister2018Optimal,Zhu2019Efficient,Zhu-general} and QGV \cite{ZhuZ20,Liu2020Efficient,Zeng2019Quantum} (cf.  \rscite{Hofmann2005Complementary,Daniel2013Minimum,Mayer2018Quantum}).

\subsection{\label{sec:QSV}Quantum state verification}
 Consider a quantum system associated with the Hilbert space $\caH$. A quantum device  is supposed to produce the target state $|\Psi \rangle$, but actually produces the $N$ states $\rho_1,\rho_2,\dots,\rho_N$ in $N$ runs. To distinguish the two situations, we can perform a random test in each run.
Each test  is determined by a test operator $E_l$, which is associated with  a two-outcome measurement of the form $\{E_l, I-E_l\}$, where $I$ is the identity operator. Here the first outcome corresponds to passing the test. To guarantee that the target state can always pass the test, the test operator $E_l$ should satisfy the condition $\langle \Psi | E_l | \Psi \rangle =1$, which means $ E_l | \Psi \rangle =| \Psi \rangle$.

If the test  $E_l$ is performed with probability $p_l$, then the performance of the above verification procedure is determined by the
 verification operator $\Omega = \sum_l p_l E_l$. Suppose $\langle \Psi |\rho |\Psi \rangle \leq 1-\varepsilon$, then the maximal probability that $\rho$ can pass each test on average is \cite{Pallister2018Optimal,Zhu2019Efficient,Zhu-general}
\begin{equation}
\max_{\langle \Psi |\rho |\Psi \rangle \leq 1-\varepsilon} \tr(\Omega \rho)=1-[1-\beta(\Omega)]\varepsilon = 1-\nu(\Omega)\varepsilon,
\end{equation}
where $\beta(\Omega)$ is the second largest eigenvalue of $\Omega$, and $\nu(\Omega)=1-\beta(\Omega)$ is the spectral gap from the maximal eigenvalue. Note that a positive spectral gap is necessary and sufficient for verifying the target state reliably, assuming that the total number of tests is not limited.

Let $\varepsilon_j = 1-\< \Psi | \rho_j |\Psi \> $ be the infidelity of the state prepared in the $j$th run and  let  $\bar{\varepsilon} = \sum_j \varepsilon_j /N$ be the average infidelity. Suppose the states $\rho_1,\rho_2,\dots,\rho_N$ prepared in the $N$ runs are independent of each other. Then  the maximal probability that these states can pass all $N$ tests is $[1-\nu(\Omega)\bar{\varepsilon}]^N$. To ensure the condition $\bar{\varepsilon}<\varepsilon $ with significant level $\delta$, the minimum number of tests required reads \cite{Pallister2018Optimal,Zhu2019Efficient,Zhu-general}
\begin{equation}
N=\left \lceil \frac{\ln \delta}{\ln [1-\nu(\Omega)\varepsilon]} \right \rceil 
\approx \frac{\ln \delta^{-1}}{\nu(\Omega)\varepsilon}.
\end{equation}
Not surprisingly, a larger spectral gap means a higher efficiency.

\subsection{\label{sec:QGV}Quantum gate verification}

Consider a quantum  device that is expected to perform the unitary transformation $\mathcal{U}$ associated with the unitary operator $U$ on $\caH$, but  actually realizes an unknown quantum process $\Lambda$. In order to verify whether this quantum process is sufficiently  close to the target unitary transformation, we need to construct a set $\scrT=\{|\psi_j\>\}_j$ of test states. In each run we randomly prepare a test state from the set $\scrT$ and apply the quantum process $\Lambda$. Then we verify whether the output state $\Lambda(\rho_j)$  is sufficiently close  to the target output state $\mathcal{U}(\rho_j)=U\rho_j U^\dag$ by virtue of QSV as described in \sref{sec:QSV}, where $\rho_j=|\psi_j\>\<\psi_j|$ \cite{ZhuZ20,Liu2020Efficient}. By construction, the target unitary transformation can always pass each test.

Suppose the test state $|\psi_j\>$ is chosen with probability $p_j>0$; denote the verification operator for the output state $\mathcal{U}(\rho_j)$ by $\Omega_j$. Then the average probability that the process $\Lambda$ can pass each test reads \cite{ZhuZ20}
\begin{equation}
\sum_{j} p_j  \tr [\Omega_j \Lambda(\rho_j)]. \label{pass_test}
\end{equation}
The target unitary transformation $\caU$ can be verified reliably if only $\caU$ can pass each test with certainty. To clarify this condition, we need to introduce additional terminology.  Let $\nu_j$ be the spectral gap of $\Omega_j$. The test state $|\psi_j\>$ is effective if  $\nu_j>0$; the set of effective test states is denoted by $\scrT_{\eff}$. The verification protocol is \emph{ordinary} if $\nu_j>0$ for each $j$, in which case every test state is effective, so that $\scrT_{\eff}=\scrT$. Otherwise, the verification protocol is \emph{extraordinary}.

A set $\scrT=\{|\psi_j\> \}_j$ in $\caH$ can \emph{identify} the unitary transformation $\caU$ if the condition 
\begin{align}\label{eq:LambdaU}
\Lambda(|\psi_j\>\<\psi_j|)=\caU(|\psi_j\>\<\psi_j|),\quad \forall j
\end{align}
implies that $\Lambda=\caU$, that is, 
\begin{align}
\Lambda(\rho)=\caU(\rho), \quad \forall \rho \in \scrD(\caH),
\end{align}
where $\scrD(\caH)$ denotes the set of all density operators on the Hilbert space $\caH$.
In this case, the set $\scrT$ is referred to as an \emph{identification set} (IS). It turns out the set $\scrT$ can identify $\caU$ iff it can identify any other unitary transformation on $\caH$ \cite{Mayer2018Quantum}, so it is not necessary to refer to a specific unitary transformation. The significance of ISs to QGV is manifested in the following lemma. Further discussions on ISs will be presented in \sref{sec:MIS}.
\begin{lemma}\label{lem:QGVreliable}
	If the unitary transformation $\caU$ can be verified reliably by a protocol based on the set $\scrT=\{|\psi_j\>\}_j$ of test states, then $\scrT$ is an IS. If the set $\scrT_{\eff}$ of effective test states is an IS, then  the unitary transformation $\caU$ can be verified reliably. If the verification protocol is ordinary, then $\caU$ can be verified reliably iff $\scrT$ is an IS. 
\end{lemma}

\begin{proof}
	By construction, $\caU$  can pass each test with certainty, so any quantum process $\Lambda$ that satisfies the condition in \eref{eq:LambdaU} can also pass each test with certainty. Suppose $\caU$ can be verified reliably. Then only $\caU$ can pass each test with certainty, which implies that $\Lambda=\caU$ when \eref{eq:LambdaU} holds.  Therefore, $\scrT$ is an IS.

	Conversely, if a quantum process $\Lambda$ can pass each test with certainty, then  we have $\tr [\Omega_j \Lambda(|\psi_j\>\<\psi_j|)]=1$ for each $|\psi_j\>\in \scrT$, which implies that 
	\begin{align}\label{eq:LambdaUeff}
	\Lambda(|\psi_j\>\<\psi_j|)=\caU(|\psi_j\>\<\psi_j|),\quad \forall |\psi_j\>\in \scrT_{\eff},
	\end{align}
	given that $|\psi_j\>\in \scrT_{\eff}$ iff $\nu_j>0$.  Now suppose the set $\scrT_{\eff}$ of effective test states is an IS, then \eref{eq:LambdaUeff} implies that $\Lambda=\caU$. Therefore, only the target unitary transformation $\caU$ can pass each test with certainty, which means $\caU$ can be verified reliably. 
	
If the verification protocol is ordinary, then 	$\scrT_{\eff}=\scrT$, so the last statement in \lref{lem:QGVreliable} follows from the first two statements. 
\end{proof}

The sample complexity of QGV has been analyzed in \rscite{ZhuZ20,Liu2020Efficient,Zeng2019Quantum} based on the idea of channel-state duality, but the details are not necessary to the current study. It turns out the verification of the unitary transformation $\caU$ is closely tied to the verification of its Choi state, especially when the verification protocol is balanced, which means $\sum_j p_j\rho_j=I/d$ \cite{ZhuZ20}. However,  verification protocols with minimal settings are in general not balanced as we shall see later. This observation shows that some important features in QGV do not have natural analogs in QSV and deserve further studies.

\section{\label{sec:BipartitePureState}Verification of  bipartite pure states with minimal settings}

Given a bipartite or multipartite pure state $|\Psi\rangle$, how many measurement settings are necessary to verify $|\Psi\rangle$ reliably? This problem is trivial if we can perform arbitrary entangling measurements, in which case one setting is enough. Unfortunately, it is not easy to realize entangling measurements  in practice, so  here we focus on verification protocols based on nonadaptive local projective measurements, which are amenable to experimental realization. This is a fundamental problem in the study of QSV that is of practical interest. However, it is in general very difficult to solve such an optimization problem if not impossible given that the potential choices of measurement settings are countless.
Even in the bipartite case, this problem has not been solved in the literature, although it is known that any bipartite pure state can be verified by two distinct tests based on adaptive local projective measurements \cite{LiHZ19}. Note that one test based on adaptive local projective measurements may entail many different measurement settings, so the result presented in \rcite{LiHZ19} does not resolve the current  problem under consideration.

Here we show that any bipartite pure state can be verified by at most two measurement settings, thereby resolving the minimal-setting problem in the bipartite scenario completely. 
\begin{theorem}\label{theorem:bipartite fewest settings}
	Every bipartite pure product state can be verified by one measurement setting. 	Every bipartite pure entangled state can be verified by two  measurement settings, but not one  measurement setting. 
\end{theorem}

\begin{proof}
Suppose the  bipartite system is associated with the bipartite Hilbert space $\caH_\rmA\otimes \caH_\rmB$ of dimension $d_\rmA\otimes d_\rmB$.
In the Schmidt basis, any bipartite pure state in  $\caH_\rmA\otimes \caH_\rmB$ can be written as 
\begin{equation}
|\Psi\>=\sum_{j=0}^{r-1} \lambda_j |jj\>,
\end{equation}	
where $r=\min\{d_\rmA,d_\rmB\}$, and $\lambda_j$ are the Schmidt coefficients of $|\Psi\>$ arranged in nonincreasing order. 

If $|\Psi\rangle$ is a product state, then $\lambda_j=\delta_{j0}$ and $|\Psi\rangle =|00\rangle$. In this case $|\Psi\rangle$ can be verified by a verification protocol composed of the single test  $P_0=|\Psi\>\<\Psi|=|00\> \<00|$. In addition, $P_0$ can be realized by one measurement setting, that is, the projective measurement onto the Schmidt basis.

If $|\Psi\>$ is entangled, then it cannot be verified by one measurement setting based on a nonadaptive local projective measurement because the pass eigenspace of any such verification operator has dimension at least 2, which means the spectral gap is zero. To prove \thref{theorem:bipartite fewest settings}, it remains to show that $|\Psi\>$ can be verified by two measurement settings. Let 
\begin{align}
P_1:=&\sum_{j=0}^{r-1} |jj\>\<jj| , \\
P_2:=&I-|u\>\<u| \otimes I + |u\>\<u| \otimes |v\>\<v|,
\end{align}
where
\begin{gather}
|u\>: = \frac{1}{\sqrt{r}} \sum_{j=0}^{r-1} |j\>, \\
|v\> := \lambda_0 |0\> + \lambda_1 |1\> +\dots + \lambda_{r-1} |r-1\>.
\end{gather}
Then $P_1$ and $P_2$ are two test projectors for $|\Psi\rangle$ that can be realized by nonadaptive local projective measurements. 
To realize $P_1$, both Alice and Bob perform projective measurements on the Schmidt basis, and the test is passed if they obtain the same outcome $j$ for $j=0,1,2,\ldots, r-1$. 
To realize $P_2$, Alice performs the two-outcome projective measurement $\{|u\>\<u|, I-|u\>\<u|\}$ and Bob performs the two-outcome projective measurement $\{|v\>\<v|, I-|v\>\<v|\}$; the test  is passed except when  Alice obtains the first outcome, while Bob obtains the second outcome.

Now we can construct a simple verification protocol for $|\Psi\>$ by performing the two tests $P_1$ and $P_2$ with probability $1/2$ each. The resulting verification operator reads $\Omega=(P_1+P_2)/2$. According  to Lemma~1 in \rcite{LiHSS21}, the spectral gap of $\Omega$ is given by 
$\nu(\Omega) = (1-\sqrt{q})/2>0 $ with
\begin{equation}
q=\|\bar{P}_1\bar{P}_2\bar{P}_1\|=\Bigl\|\frac{r-1}{r}\bar{P}_1\Bigr\|=\frac{r-1}{r},
\end{equation}
where $\bar{P}_j = P_j - |\Psi\>\<\Psi|$ for $j=1,2$.
Therefore,  $|\Psi\>$ can be verified by the strategy $\Omega$, which can be realized by two measurement settings based on nonadaptive local projective measurements. 
\end{proof}

\section{\label{sec:UMinS}Verification of unitary transformations with minimal settings}

In this section we explore verification protocols of unitary transformations with minimal settings. In addition we  introduce the concept of entanglement-free verification and clarify its connection with minimal-setting verification. Verification of bipartite unitaries is then discussed in more detail.

\subsection{\label{sec:MIS}Minimal identification sets}

Recall that a set of pure states $\scrT=\{|\psi_j\> \}_j$ in $\caH$ is an IS if it can identify unitary transformations on $\caH$  (cf. \sref{sec:QGV}) \cite{Mayer2018Quantum}. Here we are particularly interested in ISs with as few elements as possible. 
 The set $\scrT$ is a \emph{minimal identification set} (MIS) if, in addition, any proper subset is not an IS. MISs are crucial to constructing verification protocols for unitary transformations with minimal settings.

To understand the properties of ISs and MISs, we need to introduce several additional concepts. 
A set  of pure states $\scrT=\{|\psi_j\> \}_j$ in $\caH$ is a spanning set if it spans $\caH$; it is a basis if it is a spanning set that is also linearly independent. 
The \emph{transition graph} of the set $\scrT$ is a graph whose vertices are in one-to-one correspondence with the states $|\psi_j\>$; two vertices $j,k$ are adjacent if $\<\psi_j|\psi_k\>\neq 0$. The set $\scrT$ is connected if its transition graph is connected; note that here the definition is different from the usual definition in topology.
The set is a \emph{connected spanning  set} if it is a spanning set  that is connected; the set $\scrT$ is a connected linearly independent set (CLIS) if it  is a linearly independent set that is connected. A connected basis is a CLIS that is also a connected spanning set. 
By definition a CLIS can contain at most $d$ states, where $d$ is the dimension of $\caH$. Suppose the set $\scrT$  is nonempty; then  a  CLIS contained in $\scrT$ is maximal if it is not contained in any other CLIS contained in $\scrT$. Note that each state in $\scrT$ is contained in at least one maximal CLIS. In particular, $\scrT$ contains at least one maximal CLIS as a subset.

The following result proved in \rcite{Mayer2018Quantum} clarifies the conditions under which a set of  pure states can identify  unitary transformations on $\caH$. 
\begin{lemma}\label{lem:IS}
A set of pure states  in $\mathcal{H}$ is an IS iff it is a connected spanning set. 
\end{lemma}

By \lref{lem:IS}, at least $d$ test states are required to identify unitaries on $\caH$.  
To saturate the lower bound $d$, the test states must form a connected  basis. 
\begin{lemma}\label{lemma:MIS}
A set  of pure states in $\mathcal{H}$ is a MIS iff it is a connected basis. 
\end{lemma}
\Lref{lemma:MIS} clarifies the properties of MISs; it is a simple corollary of \lref{lem:IS} above and \lsref{lem:CLISmax} and \ref{lem:SpanBasis} below, which are proved in 
Appendix~\ref{app:lem:SpanBasis}. 

\begin{lemma}\label{lem:CLISmax}
	Suppose $\scrT$ is a connected spanning set in $\caH$. Then any maximal CLIS contained in $\scrT$ is a connected basis. 
\end{lemma}

\begin{lemma}\label{lem:SpanBasis}
	Every connected spanning set in $\caH$ contains a subset that forms a connected basis. Every set in $\caH$ that contains a connected spanning subset is a connected spanning set. 
\end{lemma}
Suppose $\scrT$ is a connected spanning set that is composed of $k$ pure states. As an implication of  \lref{lem:SpanBasis}, $\scrT$ contains a connected spanning subset that is composed of $k'$ pure states as long as $d\leq k'\leq k$.
To illustrate the above results, here we present a connected spanning set $\scrS$  that is composed of the computational basis and one additional state \cite{Daniel2013Minimum}:
\begin{align}
\scrT=\{|j\>\}_{j=0}^{d-1} \cup \{|\varphi\>\},
\end{align} 
where
\begin{equation}\label{eq:totally rotated state}
|\varphi\> = \frac{1}{\sqrt{d}} \sum_{j=0}^{d-1}|j\>.
\end{equation}
A  connected basis contained in $\scrT$ can be constructed as follows,
\begin{align}\label{eq:MIS}
\scrS=\{|j\>\}_{j=1}^{d-1} \cup \{|\varphi\>\}. 
\end{align}
According to \lref{lemma:MIS}, $\scrS$ is also a MIS.

\subsection{Minimal-setting verification and Entanglement-free verification} \label{sec:MSV&EFV}

Let $U$ be a unitary operator on $\caH$ and $\caU$ the associated unitary transformation.   
Recall that a general verification protocol for $U$ (which means a verification protocol for $\caU$)  consists of a set of input test states and the verification protocol for the output state associated with each input state. For simplicity, here we assume that each test  state is a pure product state, and the verification protocol for each output state is based on nonadaptive local projective measurements. Such verification protocols are most amenable to experimental realization.

 We are particularly interested in the minimum number of experimental settings required to verify $U$ by ordinary verification protocols, which  is denoted by $\mu(U)$ henceforth. When extraordinary verification protocols are allowed, the minimum number is denoted by $\mu_\rme(U)$. To be specific, one experimental setting means the preparation of a pure product input state and a 
nonadaptive local projective measurement on the output state. Note that the number of experimental settings required by any verification protocol is at least the number of test states involved. In conjunction with \lsref{lem:QGVreliable} and \ref{lem:IS}, this observation implies that  
\begin{align}\label{eq:muULB}
\mu(U)\geq \mu_\rme(U)\geq d
\end{align}
for any unitary operator $U$ acting on a $d$-dimensional Hilbert space. For a simple noncomposite system, the two inequalities  can always be saturated, and the verification problem is trivial. In the rest of this paper we shall focus on   composite systems and consider only ordinary verification protocols, in which case it is in general highly nontrivial to determine $\mu(U)$. Although it is even more difficult to determine $\mu_\rme(U)$, our results on $\mu(U)$ provide valuable upper bounds for $\mu_\rme(U)$, which are nearly tight in the bipartite setting.

A verification protocol for $U$ is \emph{entanglement free} if all input test states and the corresponding output states (after the action of $U$) are product states; in addition, all measurements are based on local projective measurements. 
An entanglement-free protocol does not generate any entanglement  in the verification procedure and hence the name. Such verification protocols are particularly appealing to both theoretical study and experimental realization.
It turns out entanglement-free verification is intimately connected to  minimal-setting verification. To clarify this point, we need to introduce some additional terminology.

Denote by $\mathrm{Prod}$ the set of pure  product states; denote by $\mathrm{Prod}(U)$ the  set of product states that remain product states after the action of  $U$:
\begin{equation}\label{eq:ProdU}
\mathrm{Prod}(U) = \{ |\psi\> \in \mathrm{Prod}\ | \ U |\psi\> \in \mathrm{Prod} \}.
\end{equation}
 The dimension of the span of the set $\mathrm{Prod}(U) $ is denoted by $d_\mathrm{Prod}(U)$:
\begin{equation}
d_\mathrm{Prod}(U) = \mathrm{dim}\ \mathrm{span} (\mathrm{Prod}(U)),
\end{equation}
which satisfies $0 \le d_\mathrm{Prod}(U) \le d$. A state $|\psi\>$ in $\caH$ satisfies the \emph{product-state constraint} associated with $U$  if $|\psi\>\in \mathrm{Prod}(U)$. A set of states satisfies the product-state constraint if it is contained in $\mathrm{Prod}(U)$, so that each state satisfies the constraint.

An entanglement-free IS (EFIS) $\scrT$ for $U$ is an IS that satisfies the product-state constraint, which implies that $\scrT\subseteq \mathrm{Prod}(U)$. 
Similarly, an entanglement-free  MIS (EFMIS)
is a MIS that satisfies the product-state constraint.  Note that the definition of an EFIS (EFMIS) depends on the specific unitary transformation under consideration, although the definition of an IS (MIS) is independent of a specific unitary transformation. The unitary operator $U$ can be verified by an entanglement-free protocol iff it admits an EFMIS, in which case $\mathrm{Prod}(U)$ contains an IS. \Lref{lem:TestStateProdU} and \thref{thm:minS-entFree} below
further clarify the connections among the product-state constraint as determined by $\mathrm{Prod}(U)$, minimal-setting verification, and entanglement-free verification. 
The proof of \Lref{lem:TestStateProdU} is presented in \aref{app:lem:TestStateProdU}.

\begin{lemma}\label{lem:TestStateProdU}
Suppose  $U$ is a unitary operator acting on a composite Hilbert space $\caH$ of dimension $d$.
Suppose $\scrT$ is the set of test states of an entanglement-free verification protocol for $U$ or an ordinary  verification protocol composed of $d$ experimental settings based on local operations. Then  $\scrT\subseteq \mathrm{Prod}(U)$. 
\end{lemma}

\begin{theorem}\label{thm:minS-entFree}
Suppose  $U$ is a unitary operator on a composite Hilbert space $\caH$ of dimension $d$. Then the following five statements are equivalent:
\begin{enumerate}

\item  $\mu(U) = d$.
	
\item $\mathrm{Prod}(U)$ is a connected spanning set. 

\item  $\mathrm{Prod}(U)$ contains a connected  basis as a subset.

\item $U$ admits an EFMIS.

\item $U$ can be verified by an entanglement-free  protocol.
\end{enumerate}	
\end{theorem}

\begin{corollary}\label{cor:minS-entFree}
Suppose  $U$ is a unitary operator on a composite Hilbert space $\caH$ of dimension $d$.	If $\mu(U) = d$ or 
if $U$ can be verified by an entanglement-free  protocol, then  $d_\mathrm{Prod}(U)= d$.
\end{corollary}
\Crref{cor:minS-entFree} is an immediate consequence of \thref{thm:minS-entFree}. 

\begin{proof}[Proof of \thref{thm:minS-entFree}]
Suppose $\mu(U)=d$. Then $U$ can be verified by an ordinary protocol composed of $d$ experimental settings that are based on local operations. Let $\scrT$ be the set of test states; then $\scrT$ forms a connected basis according to \lsref{lem:QGVreliable} and \ref{lem:IS}. In addition, $\scrT\subseteq\mathrm{Prod}(U)$ according to \lref{lem:TestStateProdU}. Therefore, $\mathrm{Prod}(U)$ is a connected spanning set according to \lref{lem:SpanBasis}, which confirms the implication $1\imply 2$.

Next, suppose $\mathrm{Prod}(U)$ is a connected spanning set. Then $\mathrm{Prod}(U)$ contains a connected basis as a subset according to  \lref{lem:SpanBasis}, which confirms the implication $2\imply 3$.

Next, suppose $\mathrm{Prod}(U)$ contains a connected basis $\scrT$. Then $\scrT$ satisfies the product-state  constraint and is a MIS according to \lref{lemma:MIS}. Therefore, $\scrT$ is an EFMIS for $U$,  which confirms the implication $3\imply 4$.

The  implication $4\imply 5$ follows from the definition, given that any EFMIS for $U$ can serve as a set of test states of an entanglement-free verification protocol.

Finally, suppose $U$ can be verified by an entanglement-free protocol; let $\scrT$ be the set of test states.  Then $\scrT$ is an IS contained in $\mathrm{Prod}(U)$ by \lref{lem:QGVreliable} and is thus a connected spanning set by \lref{lem:IS}. 
According to \lref{lem:SpanBasis}, $\scrT$ contains a connected basis $\scrS$, which enables us to construct a reliable verification protocol for $U$ using only  $d$ experimental settings. Therefore, $\mu(U)=d$, which confirms the implication $5\imply 1$ and  completes the proof of \thref{thm:minS-entFree}.
\end{proof}

\subsection{Minimal settings for verifying bipartite unitaries}

In this section we  focus on the verification of general bipartite unitaries and show that the minimum number of settings required to verify a generic bipartite unitary  grows linearly with the total dimension. 
\begin{theorem}\label{theorem:minimal settings}
Suppose $U$ is a unitary operator acting on a $d$-dimensional bipartite Hilbert space $\caH$. Then the minimum number of experimental  settings $\mu(U)$ required to verify $U$  satisfies $d \le \mu(U) \le 2d$.
\end{theorem}

\begin{proof}
The inequality $d \le \mu(U)$	follows from the general lower bound in \eref{eq:muULB}.  
To prove the upper bound $\mu(U) \le 2d$, note that the MIS $\scrS$ in \eref{eq:MIS} can serve as a set of test states; in addition, all states in  $\scrS$ are product states as long as the computational basis coincides with the standard product basis. According to  \thref{theorem:bipartite fewest settings}, the output state associated with each input state can be verified by either one or two measurement settings based on nonadaptive local projective measurements. Therefore, $\mu(U)\leq 2d$, which completes the proof of \thref{theorem:minimal settings}. 
\end{proof}

The following proposition clarifies the relation between $\mu(U)$ and $d_\mathrm{Prod}(U)$; see \aref{app:proposition:mu(U)} for a proof. 
\begin{proposition}\label{pro:muUdprod}
Let  $U$ be a unitary operator acting on a $d$-dimensional bipartite Hilbert space $\caH$.	If $d_\mathrm{Prod}(U)<d$, then 
	\begin{equation}\label{eq:mu(U)}
	\mu(U) = d_\mathrm{Prod}(U) + 2[d-d_\mathrm{Prod}(U)].
	\end{equation}
In the case $d_\mathrm{Prod}(U)=d$, we have $\mu(U)=d$ if the set $\mathrm{Prod(U)}$ is connected and $\mu(U)=d+1$ otherwise.
\end{proposition}

\section{\label{sec:TwoQubitU}Two-qubit unitaries}
In this section we discuss the basic properties of two-qubit unitaries that are relevant to  studying the minimal-setting verification and entanglement-free verification presented in the next section. Here the discussion builds on the previous works \rscite{Kraus2001Optimal,D2002Optimal}.

\subsection{\label{sec:CanonicalForm}Canonical form of two-qubit unitaries}

Let $\caH = \caH_\rmA \otimes \caH_\rmB$ be the Hilbert space associated with a two-qubit system shared by A and B. According to \rscite{Kraus2001Optimal,D2002Optimal},
any two-qubit unitary operator $U_{\rmA\rmB}$ acting on $\caH$ can be expressed as follows,
\begin{equation}\label{eq:canonical form 1}
U_{\rmA \rmB} = V_\rmA \otimes W_\rmB U \tilde{V}_\rmA \otimes \tilde{W}_\rmB,
\end{equation}
where $V_\rmA,W_\rmB,\tilde{V}_\rmA,\tilde{W}_\rmB$ are four qubit unitary operators,
\begin{equation}\label{eq:canonical form 2}
\begin{gathered}
U=U(\alpha_1, \alpha_2,\alpha_3)=\rme^{-\rmi H(\alpha_1,\alpha_2,\alpha_3)},\\
H(\alpha_1,\alpha_2,\alpha_3)=\sum_{k=1}^3\alpha_k H_k, \\ 
0 \leq |\alpha_3| \leq \alpha_2 \leq \alpha_1 \leq \pi/4, \\
H_1=\sigma_1 \otimes \sigma_1, \quad
H_2=\sigma_2 \otimes \sigma_2,\quad
H_3=\sigma_3 \otimes \sigma_3,
\end{gathered}
\end{equation}
and $\sigma_1,\sigma_2,\sigma_3$ are the three Pauli operators. 
The operator $U(\alpha_1,\alpha_2,\alpha_3)$ can further be expressed as 
\begin{equation}\label{eq:U}
U(\alpha_1,\alpha_2,\alpha_3) 
=\sum_{k=0}^3 \zeta_k \sigma_k \otimes \sigma_k,
\end{equation}
where $\sigma_0$ is the identity operator
and the coefficients $\zeta_k$ are given by
\begin{equation}\label{eq:zetak}
\begin{aligned}
\zeta_0 = \cos\alpha_1 \cos\alpha_2 \cos\alpha_3 - \mathrm{i} \sin\alpha_1 \sin\alpha_2 \sin\alpha_3,\\
\zeta_1 = \cos\alpha_1 \sin\alpha_2 \sin\alpha_3 - \mathrm{i} \sin\alpha_1 \cos\alpha_2 \cos\alpha_3,\\
\zeta_2 = \sin\alpha_1 \cos\alpha_2 \sin\alpha_3 - \mathrm{i} \cos\alpha_1 \sin\alpha_2 \cos\alpha_3,\\
\zeta_3 = \sin\alpha_1 \sin\alpha_2 \cos\alpha_3 - \mathrm{i} \cos\alpha_1 \cos\alpha_2 \sin\alpha_3.
\end{aligned}
\end{equation}

According to the equation
\begin{align}
\sigma_3^\rmA  U(\alpha_1,\alpha_2,-\alpha_3)\sigma_3^\rmA&=U(-\alpha_1,-\alpha_2,-\alpha_3)\nonumber\\
&=U^*(\alpha_1,\alpha_2,\alpha_3),
\end{align}
$U(\alpha_1,\alpha_2,-\alpha_3)$ is equivalent to $U^*(\alpha_1,\alpha_2,\alpha_3)$. Therefore, any two-qubit unitary operator is equivalent to $U(\alpha_1,\alpha_2,\alpha_3)$ or $U^*(\alpha_1,\alpha_2,\alpha_3)$ with 
\begin{equation}\label{eq:para range}
0 \leq \alpha_3 \leq \alpha_2 \leq \alpha_1 \leq \pi/4.
\end{equation}
Since most quantities we are interested in, such as Schmidt coefficients and the minimum number of experimental settings, are invariant under local unitary transformations and complex conjugation, so we can focus on  $U(\alpha_1,\alpha_2,\alpha_3)$ with the parameter range in \eref{eq:para range} in the following discussion.

\subsection{\label{sec:Schmidt}Schmidt coefficients of two-qubit unitaries}

To further clarify the properties of two-qubit unitary operators, we need to find suitable invariants. Given a two-qubit unitary operator $U$ acting on the Hilbert space $\caH = \caH_\rmA \otimes \caH_\rmB$, 
 its Choi state 
\begin{equation}\label{eq:PsiU}
|\Psi_U\> := U |\Phi\>_{\rmA\rmA'} \otimes |\Phi\>_{\rmB\rmB'}
\end{equation}
 is a four-qubit pure state on $\caH \otimes \caH$, where
 \begin{align}
|\Phi\>_{\rmA\rmA'}=\frac{1}{\sqrt{2}} \sum_k |k\>_{\rmA}|k\>_{\rmA'},\quad |\Phi\>_{\rmB\rmB'}=\frac{1}{\sqrt{2}} \sum_k |k\>_{\rmB}|k\>_{\rmB'}
 \end{align}
are two-qubit maximally entangled states shared by parties $\rmA\rmA'$ and  $\rmB\rmB'$, respectively. The Schmidt coefficients (rank) of $U$ are defined as the Schmidt coefficients (rank) of $|\Psi_U\>$ with respect to the partition between $\rmA\rmA'$ and $\rmB\rmB'$.  Note that the Schmidt coefficients and  Schmidt rank of $U$ are invariant under local unitary transformations.

Let 
\begin{align}
|\tilde{\Phi}_k\>=\sigma_k \otimes I |\Phi\>, \quad  k=0,1,2,3.
\end{align}
Then the set
$\{ |\tilde{\Phi}_k\> \}_{k=0}^3$ forms a Bell basis, which is equivalent to the magic basis \cite{PhysRevLett.78.5022} up to overall phase factors. When $U=U(\alpha_1,\alpha_2,\alpha_3)$ is the canonical two-qubit unitary  defined in \sref{sec:CanonicalForm},  by virtue of \eref{eq:U}, the Choi state $|\Psi_U\>$ can be expressed as 
\begin{align}
|\Psi_U\> 
&= \sum_{k=0}^3 \zeta_k |\tilde{\Phi}_k\>_{\rmA\rmA'} \otimes |\tilde{\Phi}_k\>_{\rmB\rmB'}.
\end{align} 
Now it is clear that the Schmidt coefficients of $|\Psi_U\>$ with respect to the partition between $\rmA\rmA'$ and $\rmB\rmB'$ are $|\zeta_k|$  for $k=0,1,2,3$, where $\zeta_k$ are given in 	\eref{eq:zetak}. Therefore, the two-qubit unitary $U(\alpha_1,\alpha_2,\alpha_3)$ has Schmidt coefficients $|\zeta_k|$  for $k=0,1,2,3$, which satisfy the following normalization condition:
\begin{equation}\label{eq:SchmidtCoeffNorm}
|\zeta_0|^2+|\zeta_1|^2+|\zeta_2|^2+|\zeta_3|^2=1.
\end{equation}
Note that $U^*(\alpha_1,\alpha_2,\alpha_3)$ and $U(\alpha_1,\alpha_2,\alpha_3)$ have the same Schmidt coefficients and  Schmidt rank. So we can focus on  the parameter range in \eref{eq:para range}  when studying the Schmidt coefficients and Schmidt rank of $U(\alpha_1,\alpha_2,\alpha_3)$.

The Schmidt rank of $U(\alpha_1,\alpha_2,\alpha_3)$ is determined in  \rcite{D2002Optimal} as reproduced in the following lemma, which can also be verified directly by virtue of \eref{eq:zetak}.
\begin{lemma}\label{lem:SchmidtRank}
	Suppose $0 \leq \alpha_3 \leq \alpha_2 \leq \alpha_1 \leq \pi/4$. Then  the Schmidt rank of $U(\alpha_1,\alpha_2,\alpha_3)$ is 1 if $\alpha_1=\alpha_2=\alpha_3=0$, is 2 if $\alpha_1>0$ and $\alpha_2=\alpha_3=0$, and is 4 if  $\alpha_1\geq \alpha_2>0$. 
\end{lemma}

\begin{figure}[t]
	\centering
	\includegraphics[scale=0.17]{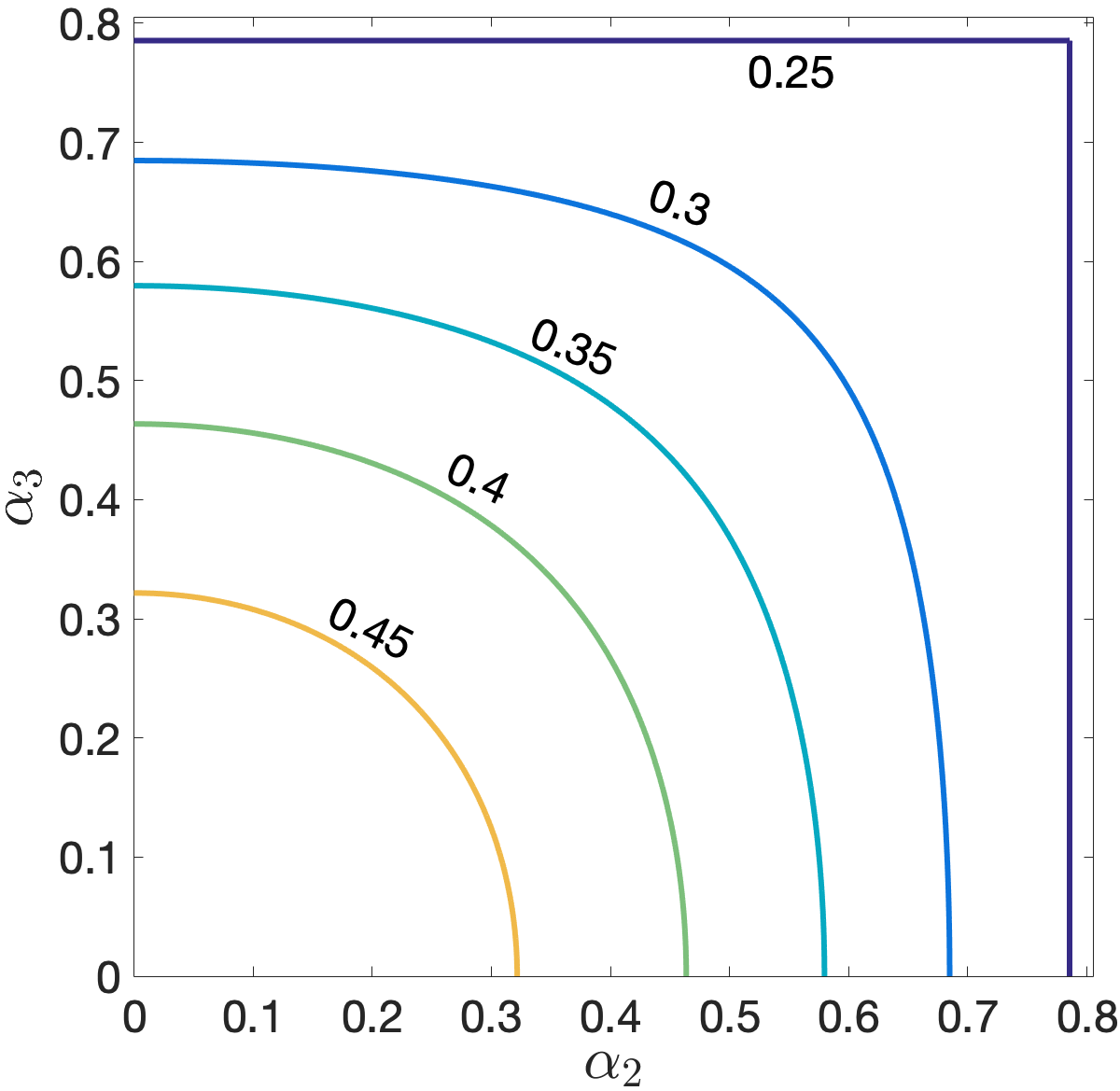}
	\caption{Contour plot of $|\zeta_0|^2$ in the plane of $\alpha_2-\alpha_3$, where  $|\zeta_0|$ is the largest Schmidt coefficient of $U(\alpha_1=\pi/4,\alpha_2,\alpha_3)$.  The other three Schmidt coefficients are determined by $|\zeta_0|^2$ according to \eref{eq:zetakSqSpecial}. All unitaries corresponding to a given contour line share the same Schmidt coefficients. 	
	}
	\label{fig:zeta0contour}
\end{figure}

The properties of Schmidt coefficients of two-qubit unitaries are summarized in \lsref{lem:SchmidtCoeffOrder}-\ref{lem:SchmidtCoeffSame} and \crref{cor:SchmidtCoeffSame} below, which are proved in \aref{app:lem:SchmidtCoeff}. 

\begin{lemma}\label{lem:SchmidtCoeffOrder}
Suppose $0 \leq \alpha_3 \leq \alpha_2 \leq \alpha_1 \leq \pi/4$. Then the  Schmidt coefficients of 
 $U(\alpha_1,\alpha_2,\alpha_3)$ satisfy the following relation:
	\begin{equation}\label{eq:SchmidtCoeffOrder}
	|\zeta_0| \ge |\zeta_1| \ge |\zeta_2| \ge |\zeta_3|\geq 0.
	\end{equation}
The first inequality saturates iff $\alpha_1 = \pi/4$; the second inequality saturates iff $\alpha_2=\alpha_1$; the third inequality saturates iff $\alpha_1=\frac{\pi}{4}$ or $\alpha_3=\alpha_2$; and the last inequality saturates iff $\alpha_2=\alpha_3=0$. 
\end{lemma}

\begin{lemma}\label{lem:SchmidtCoeffSpecial}
	Suppose $0 \leq \alpha_3 \leq \alpha_2 \leq \alpha_1 \leq \pi/4$. Then  the  four Schmidt coefficients of $U(\alpha_1,\alpha_2,\alpha_3)$  satisfy   $|\zeta_0| > |\zeta_1| = |\zeta_2| = |\zeta_3|>0$ iff $0<\alpha_3=\alpha_2=\alpha_1<\pi/4$.
\end{lemma}

When $\alpha_2=\alpha_1=\pi/4$, all Schmidt coefficients of the unitary operator $U(\alpha_1,\alpha_2,\alpha_3)$ are equal to $1/2$ irrespective of the value of $\alpha_3$ [cf. \eref{eq:zetak}]. Such coincidence can also occur when  $\alpha_1=\pi/4$ and $\alpha_2\leq\pi/4$, 
in which case we have
\begin{equation}\label{eq:zetakSqSpecial}
\begin{aligned}
|\zeta_0|^2=|\zeta_1|^2=\frac{1}{4}[1+\cos(2\alpha_2)\cos(2\alpha_3)],\\ 
|\zeta_2|^2=|\zeta_3|^2=\frac{1}{4}[1-\cos(2\alpha_2)\cos(2\alpha_3)],
\end{aligned}
\end{equation}
so all Schmidt coefficients of  $U(\alpha_1,\alpha_2,\alpha_3)$ are completely determined by the product $\cos(2\alpha_2)\cos(2\alpha_3)$ or any given Schmidt coefficient, as illustrated in \fref{fig:zeta0contour}. 
A specific choice of two inequivalent unitary operators with the same  Schmidt coefficients is shown in \aref{app:eg same coeff}.
 On the other hand, the following lemma shows that such coincidence of Schmidt coefficients cannot occur when $\alpha_1<\pi/4$. 
\begin{lemma}\label{lem:SchmidtCoeffSame}
Suppose $0 \leq \alpha_3 \leq \alpha_2 \leq \alpha_1 \leq \pi/4$ and $0 \leq \alpha_3' \leq \alpha_2' \leq \alpha_1' \leq \pi/4$. Then $U(\alpha_1,\alpha_2,\alpha_3)$ and $U(\alpha_1',\alpha_2',\alpha_3')$ have the same Schmidt coefficients iff one of the following two conditions holds,
\begin{gather}
\alpha_1=\alpha_1',\quad \alpha_2=\alpha_2',\quad \alpha_3=\alpha_3'; \label{eq:SchmidtCoeffSameCon1}\\ 
\alpha_1=\alpha_1'=\frac{\pi}{4},\quad \cos(2\alpha_2)\cos(2\alpha_3)= \cos(2\alpha_2')\cos(2\alpha_3'). \label{eq:SchmidtCoeffSameCon2}
\end{gather}
\end{lemma}

\begin{corollary}\label{cor:SchmidtCoeffSame}
Suppose $U$ and $U'$ are two two-qubit unitary operators that have the same Schmidt coefficients $s_0,s_1,s_2,s_3$, which satisfy $s_0 > s_1 \ge s_2 \ge s_3$. Then $U'$ is equivalent to either $U$ or $U^*$ under local unitary transformations. In other words, $U'$ can be expressed as
\begin{equation}
U' = V_\rmA \otimes W_\rmB \tilde{U} \tilde{V}_\rmA \otimes \tilde{W}_\rmB,
\end{equation}
where $\tilde{U}=U$ or $U^*$, and $V_\rmA,W_\rmB,\tilde{V}_\rmA,\tilde{W}_\rmB$ are suitable qubit unitary operators.
\end{corollary}

\begin{figure}[t]
	\centering
	\includegraphics[scale=0.17]{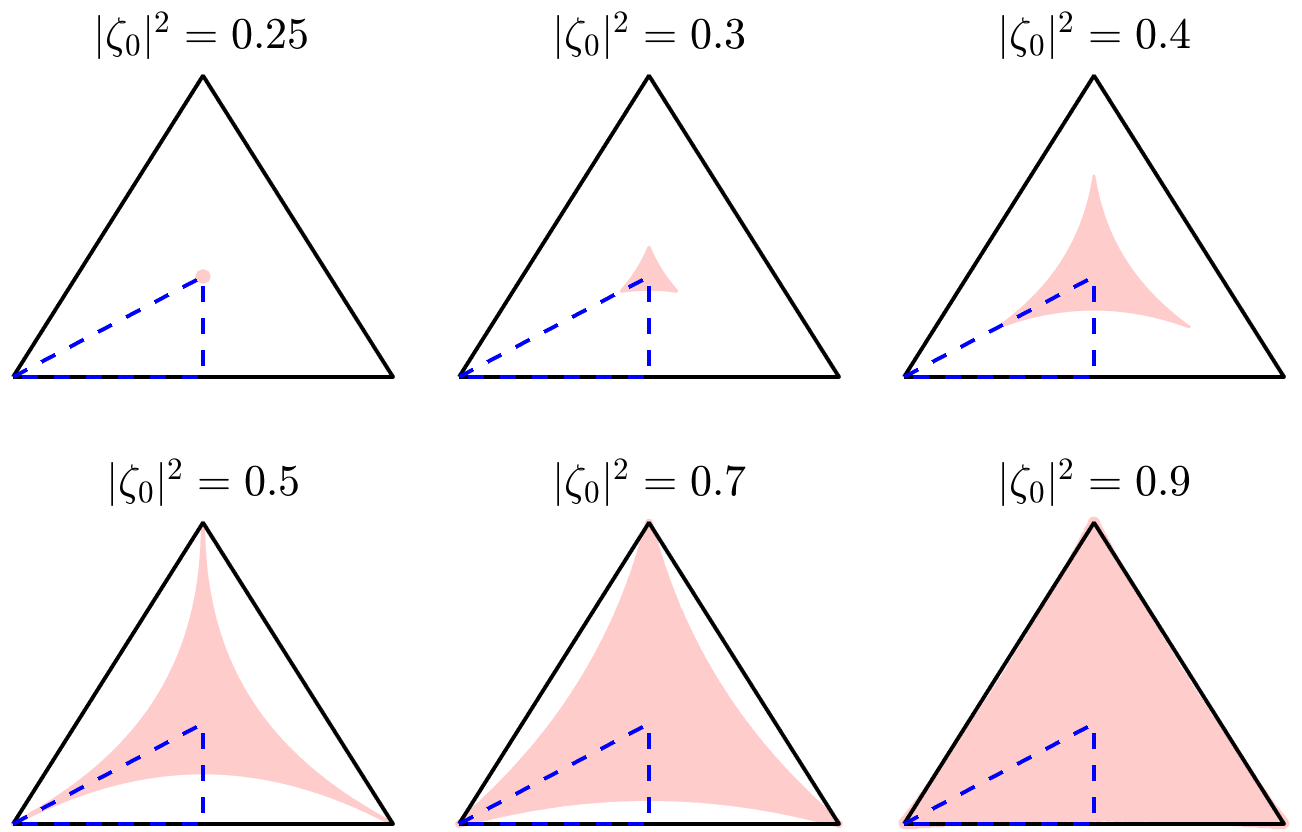}
	\caption{Accessible Schmidt coefficients of two-qubit unitaries $U(\alpha_1,\alpha_2,\alpha_3)$ for the parameter range $0 \leq \alpha_3,  \alpha_2, \alpha_1 \leq \pi/4$. The red-shaded region in each ternary diagram represents the set of  accessible points specified by the barycentric coordinate $(\xi_1,\xi_2,\xi_3)=(|\zeta_1|^2, |\zeta_2|^2, |\zeta_3|^2)/(1-|\zeta_0|^2)$, where $|\zeta_0|$ is the largest Schmidt coefficient, and $|\zeta_1|, |\zeta_2|, |\zeta_3|$ are the other three  Schmidt coefficients; cf. \eref{eq:zetak}. The left, right, and top corners of the big black triangle correspond to the coordinates $(1,0,0)$, $(0,1,0)$, and $(0,0,1)$, respectively. 
		The shaded region within each blue dashed triangle represents the set of accessible points for the smaller parameter range $0 \leq \alpha_3 \leq \alpha_2 \leq \alpha_1 \leq \pi/4$, in which case $|\zeta_1|$, $|\zeta_2|$, $|\zeta_3|$ are in nonincreasing order.}
	\label{fig:triangle_dis}
\end{figure}

The above analysis clarifies the properties of Schmidt coefficients of two-qubit unitary operators. Given the assumption $0 \leq \alpha_3 \leq \alpha_2 \leq \alpha_1 \leq \pi/4$, the Schmidt coefficients of $U(\alpha_1,\alpha_2,\alpha_3)$ must satisfy the conditions in  \esref{eq:SchmidtCoeffNorm} and \eqref{eq:SchmidtCoeffOrder}.
 However, the two conditions are not enough to guarantee the existence of a two-qubit unitary with a given set of Schmidt coefficients. 
To demonstrate this point, we can  determine the ranges of the four Schmidt coefficients of  $U(\alpha_1,\alpha_2,\alpha_3)$ by virtue of  \eref{eq:zetak}, with the result
\begin{equation}
\begin{gathered}
\frac{1}{2}\leq |\zeta_0| \leq 1, \quad 0\leq |\zeta_1|\leq \frac{1}{\sqrt{2}},\\
0\leq |\zeta_2|\leq \frac{1}{2}, \quad 0\leq |\zeta_3|\leq \frac{1}{2}. 	
\end{gathered}
\end{equation}
By contrast, the constraints in  \esref{eq:SchmidtCoeffNorm} and \eqref{eq:SchmidtCoeffOrder} alone would imply that $0\leq |\zeta_2|\leq 1/\sqrt{3}$.

To further clarify the constraints  on the Schmidt coefficients of two-qubit unitaries, it is convenient to introduce some additional variables. 
Let 
\begin{align}\label{eq:xi}
\xi_j=\frac{|\zeta_j|^2}{1-|\zeta_0|^2},\quad j=1,2,3.
\end{align}
Geometrically, $(|\zeta_0|^2, |\zeta_1|^2,  |\zeta_2|^2, |\zeta_3|^2)$ can be regarded as the barycentric coordinate of a point in a three-dimensional probability simplex according to \eref{eq:SchmidtCoeffNorm}. The accessible Schmidt coefficients correspond to a subset in the probability simplex. 
In addition, when $|\zeta_0|<1$,  $(\xi_1, \xi_2, \xi_3)$ is the barycentric coordinate of a point in a two-dimensional probability simplex, which corresponds to a normalized cross section of the three-dimensional probability simplex.

\Fref{fig:triangle_dis} illustrates the accessible region of Schmidt coefficients 
for six normalized cross sections associated with six distinct values of $|\zeta_0|$, where $|\zeta_0|$ is the largest Schmidt coefficient.
The shaded region within each blue dashed triangle represents 
the set of accessible ordered Schmidt coefficients as determined by $(\xi_1, \xi_2, \xi_3)$
for the parameter range $0 \leq \alpha_3 \leq \alpha_2 \leq \alpha_1 \leq \pi/4$. 
By contrast, the whole red-shaded region in each ternary diagram represents the set of accessible Schmidt coefficients for the larger parameter range $0 \leq \alpha_3,  \alpha_2, \alpha_1 \leq \pi/4$. In the latter case,    \eref{eq:SchmidtCoeffOrder} no longer applies, but we have 
\begin{equation}
|\zeta_0|\geq |\zeta_j| \quad j=1,2,3,
\end{equation}
so $|\zeta_0|$ is still the largest Schmidt coefficient.

\section{\label{sec:VTwoQuibtU}Verification of  two-qubit unitaries with minimal settings}

\subsection{Product-state constraint}
To construct a minimal-setting  protocol for verifying  the two-qubit unitary $U(\alpha_1,\alpha_2,\alpha_3)$, we first need to
clarify the product-state constraint, which is tied to  the set $\mathrm{Prod}(U)$ defined in \eref{eq:ProdU}.

To better understand the product-state constraint, it is instructive to consider the magic basis \cite{PhysRevLett.78.5022}, which is composed of the four maximally entangled states
\begin{equation}
\begin{aligned}
|\Phi_1 \rangle =\frac{1}{\sqrt{2}} (|00\rangle + |11\rangle), \;\;
|\Phi_2 \rangle =\frac{\mathrm{i}}{\sqrt{2}} (|00\rangle - |11\rangle),\\
|\Phi_3 \rangle =\frac{\mathrm{i}}{\sqrt{2}} (|01\rangle + |10\rangle), \; \; 
|\Phi_4 \rangle =\frac{1}{\sqrt{2}} (|01\rangle - |10\rangle).
\end{aligned}
\end{equation}
Suppose the input state $|\phi_0\rangle$ employed  has the form
$
|\phi_0\rangle = \sum_{k=1}^4 \gamma_k |\Phi_k \rangle$ with $\sum_{k=1}^4 |\gamma_k|^2 =1$. 
Then the concurrence \cite{PhysRevLett.78.5022} of the input state reads 
\begin{equation}
C(|\phi_0 \>)= \left| \sum_{k=1}^4 \gamma_k^2 \right|.
\end{equation}
After the action of  $U(\alpha_1,\alpha_2,\alpha_3)$, the output state has the expansion 
\begin{equation}
  |\phi \> = \sum_{k=1}^4 \rme^{-\rmi \lambda_k} \gamma_k |\Phi_k \> ,
\end{equation}
where
\begin{equation}
  \begin{aligned}
  &\lambda_1 = \alpha_1-\alpha_2+\alpha_3,\\
  &\lambda_2 = -\alpha_1+\alpha_2+\alpha_3,\\
  &\lambda_3 =\alpha_1+\alpha_2-\alpha_3,\\
  &\lambda_4 = -\alpha_1-\alpha_2-\alpha_3.
  \end{aligned}
\end{equation}
The concurrence of the output state reads
\begin{equation}
C(|\phi \>)= \left| \sum_{k=1}^4 \rme^{-2\rmi \lambda_k}\gamma_k^2 \right|.
\end{equation}

The product-state constraint demands $C(|\phi_0 \>)=0$ and $C(|\phi \>)=0$:
\begin{align}\label{eq:PSCmagic}
\sum_{k=1}^4 \gamma_k^2=0,\quad \sum_{k=1}^4 \rme^{-2\rmi \lambda_k}\gamma_k^2=0.
\end{align}
When $0 < \alpha_1+\alpha_2 < \pi/2$, \eref{eq:PSCmagic} is equivalent to the following equations:
\begin{equation}\label{eq:product-state-constraint}
\begin{aligned}
\gamma_3^2 &= r_{31} \gamma_1^2 + r_{32} \gamma_2^2,\quad 
\gamma_4^2 = r_{41} \gamma_1^2 + r_{42} \gamma_2^2, \\
r_{31}
&=\exp[\mathrm{i} (2\alpha_2-2\alpha_3+\pi)] \frac{\sin(2\alpha_1+2\alpha_3)}{\sin(2\alpha_1+2\alpha_2)},\\
r_{32}
&=\exp[\mathrm{i} (2\alpha_1-2\alpha_3+\pi)] \frac{\sin(2\alpha_2+2\alpha_3)}{\sin(2\alpha_1+2\alpha_2)},\\
r_{41}
&=\exp[\mathrm{i} (-2\alpha_1-2\alpha_3+\pi)] \frac{\sin(2\alpha_2-2\alpha_3)}{\sin(2\alpha_1+2\alpha_2)},\\
r_{42}
&=\exp[\mathrm{i} (-2\alpha_2-2\alpha_3+\pi)] \frac{\sin(2\alpha_1-2\alpha_3)}{\sin(2\alpha_1+2\alpha_2)}. 
\end{aligned}
\end{equation}
If  the product-state constraint holds, then $\gamma_3^2$ and $\gamma_4^2$ are completely determined by $\gamma_1$ and $\gamma_2$. Taking into account the normalization condition $\sum_{k=1}^4 |\gamma_k|^2 =1$ and ignoring the overall phase factors, we can deduce that there are in general two free real parameters. 

When $\alpha_1+\alpha_2=0$ or $\alpha_1+\alpha_2 = \pi/2$, \eref{eq:product-state-constraint} does not apply, in which case it is more convenient to consider the product-state constraint in the computational basis. Now any  two-qubit pure product state can be expressed as
\begin{equation}
|\phi_0 \rangle = 
\begin{pmatrix}{}
a_1 \\
a_2
\end{pmatrix}
\otimes
\begin{pmatrix}
b_1 \\
b_2
\end{pmatrix}=
\begin{pmatrix}
a_1 b_1 \\
a_1 b_2 \\
a_2 b_1 \\
a_2 b_2
\end{pmatrix}.
\end{equation}
After the action of  $U(\alpha_1,\alpha_2,\alpha_3)$, the output state reads
\begin{equation}\label{eq:outputStateCB}
|\phi \rangle = U |\phi_0 \rangle = 
\begin{pmatrix}
c_1 \\
c_2 \\
c_3 \\
c_4
\end{pmatrix},
\end{equation}
where
\begin{equation}
\begin{aligned}
c_1 = (\zeta_0+\zeta_3)a_1 b_1+(\zeta_1-\zeta_2)a_2 b_2, \\
c_2 = (\zeta_0-\zeta_3)a_1 b_2+(\zeta_1+\zeta_2)a_2 b_1, \\
c_3 = (\zeta_0-\zeta_3)a_2 b_1+(\zeta_1+\zeta_2)a_1 b_2, \\
c_4 = (\zeta_0+\zeta_3)a_2 b_2+(\zeta_1-\zeta_2)a_1 b_1, 
\end{aligned}
\end{equation}
and $\zeta_k$ for $k=0,1,2,3$ are defined in \eref{eq:zetak}. According to \rcite{PhysRevLett.78.5022},
the concurrence $C$ of the output state reads
\begin{align}\label{eq:concurrence}
C(|\phi \>) =  2 | c_1c_4-c_2c_3 |.
\end{align}
To satisfy the product-state constraint, the concurrence  $C(|\phi \>)$ should vanish, which means
\begin{equation}\label{eq:PSCstandard}
  c_1c_4-c_2c_3 = 0.
\end{equation}

\subsection{Minimal-setting and entanglement-free verification of two-qubit unitaries}

In this section we determine the minimum number of experimental settings required to verify an arbitrary two-qubit unitary and derive a simple criterion for determining  whether a general two-qubit unitary can be verified by an entanglement-free  protocol. Our main result is summarized in the following theorem. 
\begin{theorem}\label{thm:MinStwoqubitG}
	Suppose $U$ is a two-qubit unitary operator with Schmidt coefficients $s_0, s_1, s_2, s_3$	 arranged in nonincreasing order. Then 
	\begin{equation}
	\mu(U)=\begin{cases}
	5 & \mbox{if } s_0>s_1=s_2=s_3>0,\\
	4&\mbox{otherwise};
	\end{cases}
	\end{equation}
in addition, the unitary operator $U$ can be verified by an entanglement-free protocol unless $s_0>s_1=s_2=s_3>0$. 
\end{theorem}
\Thref{thm:MinStwoqubitG} is a corollary of \lref{lem:SchmidtCoeffSpecial} in \sref{sec:Schmidt} and   \thref{thm:MinStwoqubit} below. 
Define
\begin{align}\label{eq:ent-free area}
\mathcal{S}:=&\Bigl\{  (\alpha_1,\alpha_2,\alpha_3)\Big| 0 \leq \alpha_3 \leq \alpha_2 \leq \alpha_1 \leq \frac{\pi}{4} \Bigr\},\\
\caS_\rmE:=&\Bigl\{ (\alpha,\alpha,\alpha) \Big| 0 < \alpha< \frac{\pi}{4} \Bigr\},\quad
\caS_\mathrm{EF}:=\caS\setminus \caS_\rmE.
\end{align}

\begin{theorem}\label{thm:MinStwoqubit}
Suppose $0 \leq \alpha_3 \leq \alpha_2 \leq \alpha_1 \leq \pi/4$. Then 
\begin{align}\label{eq:MinStwoqubit}
\mu(U(\alpha_1,\alpha_2,\alpha_3))=\begin{cases}
4 & \mbox{if } (\alpha_1,\alpha_2,\alpha_3) \in \caS_\mathrm{EF},\\
5 & \mbox{if } (\alpha_1,\alpha_2,\alpha_3) \in \caS_\rmE. 
\end{cases}
\end{align}
$U(\alpha_1,\alpha_2,\alpha_3)$ can be verified by an entanglement-free protocol iff $(\alpha_1,\alpha_2,\alpha_3) \in \caS_\mathrm{EF}$. 
\end{theorem}

\begin{proof}
To prove \thref{thm:MinStwoqubit}, it suffices to prove \eref{eq:MinStwoqubit}, which implies 
the last statement in the theorem according to \thref{thm:minS-entFree}. To prove \eref{eq:MinStwoqubit}, we shall first construct 
a four-setting entanglement-free protocol for verifying $U(\alpha_1,\alpha_2,\alpha_3)$ when $(\alpha_1,\alpha_2,\alpha_3) \in \caS_\mathrm{EF}$. To this end we need to consider three different cases and construct an  EFMIS in each case (cf. \thref{thm:minS-entFree}).  
\begin{enumerate}
		
		\item[1.] $\alpha_1=\alpha_2=\pi/4$
		
		In this case,  according to  	\esref{eq:outputStateCB}-\eqref{eq:PSCstandard}, the product-state constraint under the computational basis reads 
		\begin{equation}
		a_1 a_2 b_1 b_2 \cos(2\alpha_3) =0. 
		\end{equation}
	In addition, $|\zeta_0|=|\zeta_1|=|\zeta_2|=|\zeta_3|=1/2$ according to \eref{eq:zetak}. So a pure product state satisfies the product-state constraint if one of the reduced states is an eigenstate of $\sigma_3$. Based on this observation we can construct an EFMIS  as follows:
  \begin{equation}\label{eq:1st-case MIS}
  \begin{aligned}
    &|\phi_1\>=|0{+}\>,\quad
    |\phi_2\>=|1{+}\>,\\ 
    &|\phi_3\>=|{-}0\>,\quad 
    |\phi_4\>=|{+}0\>,
  \end{aligned}
  \end{equation}
where $|\pm\>=\frac{1}{\sqrt{2}}(|0\>\pm |1\>)$ are the two eigenstates of $\sigma_1$.
Note that these product states 	remain as product states after the action of 	$U(\alpha_1,\alpha_2,\alpha_3)$ as expected. In addition, the transition graph of these states is connected. Therefore, 
$U(\alpha_1,\alpha_2,\alpha_3)$ can be  verified by an entanglement-free protocol based on four experimental settings, which confirms \eref{eq:MinStwoqubit}.

		\item[2.] $\alpha_1=\alpha_2=\alpha_3=0$
		
		In this case, $U(\alpha_1,\alpha_2,\alpha_3)$ is equal to the identity, so all product states satisfy the product-state constraint, and it is easy to construct an EFMIS. Actually, the EFMIS constructed in case 1 still works. Therefore,
	$U(\alpha_1,\alpha_2,\alpha_3)$ can be  verified by an entanglement-free protocol based on four experimental settings, which confirms \eref{eq:MinStwoqubit}.

		\item[3.] $\alpha_1>\alpha_3$ and  $\alpha_2<\pi/4$. 
		
	 In this case, it is more convenient to consider the magic basis. Suppose the state $|\phi_0\rangle$  has the expansion
	$|\phi_0\rangle = \sum_{k=1}^4 \gamma_k |\Phi_k \rangle$ with the normalization condition $\sum_{k=1}^4 |\gamma_k|^2 =1$.
Then the product-state constraint is satisfied if 	the coefficients $\gamma_1^2, \gamma_2^2, \gamma_3^2, \gamma_4^2$ have the form as shown in 
 \aref{app:case3-EFMIS}. Moreover, an EFMIS can be constructed as follows (in the magic basis):
		\begin{equation}\label{eq:case3-input}
		\begin{aligned}
		&|\phi_{1}\>=
		\begin{pmatrix}
		\gamma_{1} \\ \gamma_{2} \\ \gamma_{3} \\ \gamma_{4}
		\end{pmatrix},\quad 
		&|\phi_{2}\>=
		\begin{pmatrix}
		\gamma_{1} \\ -\gamma_{2} \\ \gamma_{3} \\ \gamma_{4}
		\end{pmatrix}, \\
		&|\phi_{3}\>=
		\begin{pmatrix}
		\gamma_{1} \\ \gamma_{2} \\ \gamma_{3} \\ -\gamma_{4}
		\end{pmatrix},\quad 
		&|\phi_{4}\>=
		\begin{pmatrix}
		-\gamma_{1} \\ \gamma_{2} \\ \gamma_{3} \\ -\gamma_{4}
		\end{pmatrix}.
		\end{aligned}
		\end{equation}
Therefore, 
$U(\alpha_1,\alpha_2,\alpha_3)$ can be verified by an entanglement-free protocol based on four experimental settings,  which confirms \eref{eq:MinStwoqubit}.	
	\end{enumerate}

To complete the proof of \thref{thm:MinStwoqubit}, it remains to determine $\mu(U(\alpha_1,\alpha_2,\alpha_3))$ in the case $(\alpha_1,\alpha_2,\alpha_3) \in \caS_\rmE$, which means $0<\alpha_1=\alpha_2=\alpha_3<\pi/4$.  Suppose the input state $|\phi_0\rangle$  has the expansion
$|\phi_0\rangle = \sum_{k=1}^4 \gamma_k |\Phi_k \rangle$ with $\sum_{k=1}^4 |\gamma_k|^2 =1$ in the magic basis. According to \eref{eq:product-state-constraint},  the product-state constraint amounts to the following equality:
\begin{equation}
(\gamma_1^2, \gamma_2^2, \gamma_3^2, \gamma_4^2) = (\gamma_1^2, \gamma_2^2, -\gamma_1^2- \gamma_2^2, 0),
\end{equation}
which implies that  $d_{\mathrm{Prod}}(U) = 3$. So  $U(\alpha_1,\alpha_2,\alpha_3)$ cannot be verified by an entanglement-free protocol according to \thref{thm:minS-entFree}. Nevertheless, $U(\alpha_1,\alpha_2,\alpha_3)$ can be verified by a five-setting protocol based on local operations, given that $\mu(U(\alpha_1,\alpha_2,\alpha_3))=5$ according to  \pref{pro:muUdprod}. This result
 confirms \eref{eq:MinStwoqubit} and completes the proof of \thref{thm:MinStwoqubit}. 
\end{proof}

Next, we generalize \thref{thm:MinStwoqubit} to the whole parameter range $0\le \alpha_3,\alpha_2,\alpha_1 < 2\pi$. 
Define 
\begin{align}\label{eq:ent-free area2}
\tilde{\caS}:=&\Bigl\{  (\alpha_1,\alpha_2,\alpha_3)\Big| 0 \leq \alpha_3,\alpha_2,\alpha_1 < 2\pi \Bigr\},\\
\tilde{\caS}_\rmE:=&\Bigl\{ \Bigl(\frac{\pi}{2} k_1+\frac{\pi}{4} \pm \alpha,\frac{\pi}{2} k_2+\frac{\pi}{4} \pm \alpha,\frac{\pi}{2} k_3+\frac{\pi}{4} \pm \alpha \Bigr) \Big| \nonumber \\
&  0 < \alpha < \frac{\pi}{4}, \;
k_1,k_2,k_3=0,1,2,3 \Bigr\},\\
\tilde{\caS}_\mathrm{EF}:=&\tilde{\caS} \setminus \tilde{\caS}_\rmE.
\end{align} 
The following corollary  is proved in \aref{app:cor:MinStwoqubit-general}.

\begin{corollary}\label{cor:MinStwoqubit-general}
Suppose $0\le \alpha_3,\alpha_2,\alpha_1 < 2\pi$. Then 
\begin{align}
\mu(U(\alpha_1,\alpha_2,\alpha_3))=\begin{cases}
4 & \mbox{if } (\alpha_1,\alpha_2,\alpha_3) \in \tilde{\caS}_\mathrm{EF},\\
5 & \mbox{if } (\alpha_1,\alpha_2,\alpha_3) \in \tilde{\caS}_\rmE. 
\end{cases}
\end{align}
$U(\alpha_1,\alpha_2,\alpha_3)$ can be verified by an entanglement-free protocol iff $(\alpha_1,\alpha_2,\alpha_3) \in \tilde{\caS}_\mathrm{EF}$. 
\end{corollary}

\Thref{thm:MinStwoqubit} and \crref{cor:MinStwoqubit-general} imply that generic two-qubit unitary transformations (except for a  set of measure zero) can be verified by entanglement-free protocols based on four experimental settings. In principle we can reach arbitrarily high precision as long as sufficiently many tests can be performed. 
Nevertheless, certain special unitary transformations cannot be verified by entanglement-free protocols, in which case five experimental settings are necessary. Note that the minimum number of settings is not continuous, which is expected for a discrete figure of merit. For each unitary  $U$ in the later case, we can find a nearby unitary $U'$ that can be verified by an entanglement-free protocol. In this way $U$ can be verified approximately by an entanglement-free protocol. However, the precision is limited by the entanglement infidelity between $U'$ and $U$; in addition, the target unitary transformation $U$ cannot pass all the tests with certainty. 
To enhance  the precision, we can find a better approximation to $U$, but the precision is still limited for any given approximation. Although any two-qubit unitary transformation can be verified with five measurement settings (only four settings in the generic case),  quite often the sample efficiency can be improved by increasing the number of measurement settings. The tradeoff between the sample efficiency and the number of experimental settings deserves further studies.\\

\subsection{Examples}\label{example}

In this section we present explicit EFMISs for several well-known two-qubit gates, from which entanglement-free verification protocols can be constructed immediately.

\subsubsection{CNOT}
The 
CNOT gate is equivalent to $U(\frac{\pi}{4},0,0)$ according to the following decomposition
\begin{equation}
\begin{pmatrix}
1 & 0 & 0 & 0 \\
0 & 1 & 0 & 0 \\
0 & 0 & 0 & 1 \\
0 & 0 & 1 & 0
\end{pmatrix}
= V_\rmA \otimes W_\rmB U(\frac{\pi}{4},0,0) \tilde{V}_\rmA \otimes \tilde{W}_\rmB,
\end{equation}
where
\begin{equation}
\begin{aligned}
  V_\rmA&=\frac{1}{\sqrt{2}}
  \begin{pmatrix}
  1 & 1 \\
  \rmi & -\rmi
  \end{pmatrix},&\quad 
  \tilde{V}_\rmA&=\frac{1}{\sqrt{2}}
  \begin{pmatrix}
  1 & 1 \\
  1 & -1
  \end{pmatrix}, \\
  W_\rmB&=\frac{1}{\sqrt{2}}
  \begin{pmatrix}
  1 & \rmi \\
  -\rmi & -1
  \end{pmatrix},&\quad
  \tilde{W}_\rmB&=
  \begin{pmatrix}
  1 & 0 \\
  0 & -1
  \end{pmatrix}.
\end{aligned}
\end{equation}
To construct an entanglement-free protocol for verifying the CNOT gate, it suffices to construct an EFMIS. To this end, we can first construct an EFMIS for
$U(\frac{\pi}{4},0,0)$ and then apply a suitable local unitary transformation, although it is easy to construct an EFMIS for the CNOT gate directly. 
According to \esref{eq:outputStateCB}-\eqref{eq:PSCstandard}, the product-state constraint for  $U(\frac{\pi}{4},0,0)$ under the computational basis can be expressed as
\begin{equation}
  (a_1^2-a_2^2)(b_1^2-b_2^2)=0. 
\end{equation}
A product state satisfies the constraint iff
one of the reduced states is an eigenstate of $\sigma_1$. Based on this observation, an EFMIS can be constructed  as
\begin{equation}\label{eq:input-CNOT0}
\begin{aligned}
  |\phi_1\>&=|0+\>, &\quad 
  |\phi_2\>&=|1+\>,\\ 
  |\phi_3\>&=|{-}0\>,&\quad 
  |\phi_4\>&=|{+}0\>,
\end{aligned}
\end{equation}
where $|\pm\>=(|0\>\pm |1\>)/\sqrt{2}$ are the two eigenstates of $\sigma_1$. By multiplying the local unitary operator $(\tilde{V}_\rmA \otimes \tilde{W}_\rmB)^\dag$, we can construct an EFMIS for the CNOT gate as
\begin{equation}\label{eq:input-CNOT}
\begin{aligned}
  |\tilde{\phi}_1\>&=|{+}{-}\>, &\quad
  |\tilde{\phi}_2\>&=|{-}{-}\>,\\
  |\tilde{\phi}_3\>&=|10\>,&\quad
  |\tilde{\phi}_4\>&=|00\>.
\end{aligned}
\end{equation}

\subsubsection{CZ}

The CZ gate is equivalent to the CNOT gate
according to the identity
\begin{align}
\mathrm{CZ}=(I\otimes H) \mathrm{CNOT} (I\otimes H),
\end{align}
where $H$ is the Hadamard gate.
Therefore, any EFMIS for the CNOT gate can be turned into an EFMIS for the CZ gate by simply applying the local unitary operator $I\otimes H$. 
For example, one EFMIS for the CZ gate can be constructed by applying 
$I\otimes H$ to the states in \eref{eq:input-CNOT}, which yields
\begin{equation}\label{eq:EFMIS-CZ}
\begin{aligned}
  |\phi_1\>&=|{+}1\>, &\quad 
  |\phi_2\>&=|{-}1\>,\\ 
  |\phi_3\>&=|1{+}\>,&\quad
  |\phi_4\>&=|0{+}\>.
\end{aligned}
\end{equation}

\subsubsection{C-Phase}

The C-Phase gate with nontrivial phase $0<\varphi<2\pi$ reads
\begin{equation}
\begin{pmatrix}
1 & 0 & 0 & 0 \\
0 & 1 & 0 & 0 \\
0 & 0 & 1 & 0 \\
0 & 0 & 0 & \rme^{\rmi \varphi}
\end{pmatrix}.
\end{equation}
The conjugate of the C-Phase gate is equivalent to $U\bigl(\frac{\varphi}{4},0,0\bigr)$ according to the following decomposition
\begin{equation}
\begin{pmatrix}
1 & 0 & 0 & 0 \\
0 & 1 & 0 & 0 \\
0 & 0 & 1 & 0 \\
0 & 0 & 0 & \rme^{-\rmi \varphi}
\end{pmatrix}
= V_\rmA \otimes W_\rmB U\Bigl(\frac{\varphi}{4},0,0\Bigr) \tilde{V}_\rmA \otimes \tilde{W}_\rmB,
\end{equation}
where
\begin{equation}
\begin{aligned}
V_\rmA&=\frac{1}{\sqrt{2}}
\begin{pmatrix}
1 & 1 \\
-\rme^{-\rmi \frac{\varphi}{2}} & \rme^{-\rmi \frac{\varphi}{2}}
\end{pmatrix},&\!
\tilde{V}_\rmA&=\frac{1}{\sqrt{2}}
\begin{pmatrix}
1 & -1 \\
1 & 1
\end{pmatrix},\\
W_\rmB&=\frac{1}{\sqrt{2}}
\begin{pmatrix}
\rme^{\rmi \frac{\varphi}{4}} & \rme^{\rmi \frac{\varphi}{4}} \\
\rme^{-\rmi \frac{\varphi}{4}} & -\rme^{-\rmi \frac{\varphi}{4}}
\end{pmatrix},&\!
\tilde{W}_\rmB&=\frac{1}{\sqrt{2}}
\begin{pmatrix}
1 & 1 \\
1 & -1
\end{pmatrix}.
\end{aligned}
\end{equation}

According to  \esref{eq:outputStateCB}-\eqref{eq:PSCstandard}, the product-state constraint for $U(\frac{\varphi}{4},0,0)$ under the computational basis can be expressed as
\begin{equation}
  (a_1^2-a_2^2)(b_1^2-b_2^2)\sin\frac{\varphi}{2}=0. 
\end{equation}
A product state satisfies the constraint if
one of the reduced states is an eigenstate of $\sigma_1$. So the states in \eref{eq:input-CNOT0} also form an EFMIS for $U(\frac{\varphi}{4},0,0)$. 
By applying the local unitary operator $(\tilde{V}_\rmA \otimes \tilde{W}_\rmB)^\dag$, we can construct an EFMIS for the C-Phase gate (and its conjugate) as
\begin{equation}
\begin{aligned}
  |\phi_1\>&=|{-}0\>, &\quad 
  |\phi_2\>&=|{+}0\>,\\ 
  |\phi_3\>&=-|1{+}\>,&\quad
  |\phi_4\>&=|0{+}\>.
\end{aligned}
\end{equation}
Note that this EFMIS applies to the C-Phase gate with an arbitrary phase. Incidentally, the four states in \eref{eq:EFMIS-CZ} also form an  EFMIS for the C-Phase gate with an arbitrary phase.

\subsubsection{SWAP}
The SWAP gate is equal to $U(\frac{\pi}{4},\frac{\pi}{4},\frac{\pi}{4})$ up to an overall phase factor according to the following identity
\begin{equation}
\begin{pmatrix}
1 & 0 & 0 & 0 \\
0 & 0 & 1 & 0 \\
0 & 1 & 0 & 0 \\
0 & 0 & 0 & 1
\end{pmatrix}
=\frac{1+\rmi}{\sqrt{2}} U(\frac{\pi}{4},\frac{\pi}{4},\frac{\pi}{4}).
\end{equation}
Thanks to this identity, the EFMIS for $U(\frac{\pi}{4},\frac{\pi}{4},\frac{\pi}{4})$ presented in \eref{eq:1st-case MIS} is also an EFMIS for the SWAP gate. In addition, any product state satisfies the product-state constraint, so any MIS composed of product states is an EFMIS for the SWAP gate.

\section{\label{sec:summary}Summary}

We studied systematically QSV and QGV with a focus on the number of experimental settings based on local operations.   We showed that any bipartite pure state can be verified by only two measurement settings based on local projective measurements. The minimum number of experimental settings required to verify a bipartite unitary increases linearly with the total dimension. 
In addition, we introduced the concept of entanglement-free verification, which does not generate any entanglement in the verification procedure. 
The connection with minimal-setting verification is also clarified. 
Finally, we determined the minimum number of experimental settings required to verify each two-qubit unitary. It turns out any two-qubit unitary can be verified using at most five settings based on local operations, and a generic two-qubit unitary requires only four settings. 
In the course of study we derived a number of results on two-qubit unitaries and their Schmidt coefficients, which are of independent interest. Our work significantly promotes the current understanding
on  QSV and QGV with respect to the number of required experimental settings, which is instructive for both theoretical studies and practical applications.
In addition, our work shows that verification protocols with minimal settings are in general not balanced and thus do not have natural analogs in QSV, which reflects a key distinction between QGV and QSV that is not recognized before.
In the future it would be desirable to generalize our results to the multipartite setting.

\section*{Acknowledgments}
This work is  supported by   the National Natural Science Foundation of China (Grants No.~92165109 and No.~11875110) and  Shanghai Municipal Science and Technology Major Project (Grant No.~2019SHZDZX01).

\appendix

\section{Proofs of \lsref{lem:CLISmax} and \ref{lem:SpanBasis}}\label{app:lem:SpanBasis}

\begin{proof}[Proof of \lref{lem:CLISmax}]
Suppose on the contrary that $\scrS$ is a maximal CLIS contained in $\scrT$ and that $\scrS$ is not a basis for $\caH$. Let $\caH_1$ be the span of $\scrS$ and let $\caH_2$ be the orthogonal complement of $\caH_1$. Then $\caH_1$ and $\caH_2$ have dimensions at least one; in addition,
 $\scrT$ contains a ket $|\psi\>$ that is supported neither in $\caH_1$ nor in $\caH_2$ since otherwise $\scrT$ cannot be connected. Therefore, $\scrS\cup \{|\psi\>\}\subseteq \scrT$ is a CLIS that contains $\scrS$ as a proper subset. This contradiction completes the proof of \lref{lem:CLISmax}. 
\end{proof}

\begin{proof}[Proof of \lref{lem:SpanBasis}]
The first statement in \lref{lem:SpanBasis} follows from \lref{lem:CLISmax}; note that any maximal CLIS contained in the connected spanning set forms a connected basis. To prove the second statement, suppose $\scrT$ is a set of kets in $\caH$ and contains a connected spanning set $\scrS$. Then $\scrT$ is also a spanning set. In addition, each ket in $\scrT$ is not orthogonal to at least one ket in $\scrS$. As a consequence, the transition graph of $\scrT$ is connected given that the transition graph of $\scrS$ is connected. So $\scrT$ is itself a connected spanning set, which completes the proof of \lref{lem:SpanBasis}. 
\end{proof}

\section{Proof of \lref{lem:TestStateProdU}}\label{app:lem:TestStateProdU}

\begin{proof}
For an entanglement-free verification protocol, the conclusion follows from the very definition. So it remains to consider the case in which the verification protocol is composed of $d$ experimental settings based on local operations.  
Then  we have $d\leq |\scrT|\leq d$, where the lower bound follows from the fact that $\scrT$ is a spanning set and the upper bound follows from the fact that the number of experimental settings cannot be smaller than the number of test states. 
It follows that $|\scrT|=d$ and 
$\scrT$ is composed of  $d$ product states. In addition, the number of experimental settings is equal to the number of test states. So the output state associated with each input state in $\scrT$ is also a product state given that at least two measurement settings are required to verify an entangled output state (cf. \thref{theorem:bipartite fewest settings}). Therefore, $\scrT\subseteq \mathrm{Prod}(U)$, which completes the proof of \lref{lem:TestStateProdU}.
\end{proof}

\section{Proof of  \pref{pro:muUdprod}}\label{app:proposition:mu(U)}

\begin{proof}
To prove \eref{eq:mu(U)} in \pref{pro:muUdprod}, we shall first prove the following inequality
		\begin{equation}\label{eq:mu(U)proof}
	\mu(U) \geq d_\mathrm{Prod}(U) + 2[d-d_\mathrm{Prod}(U)].
	\end{equation}
Let  $\scrT$ be the set of test states 	of a verification protocol of $U$ that can be realized by $\mu(U)$ experimental settings. 
Then $\scrT$ is a finite spanning set (of $\caH$)  whose cardinality satisfies $d\leq |\scrT|\leq \mu(U)$.  Let $\scrT'=\mathrm{Prod}(U)\cap\scrT$ and $\scrT''=\scrT\setminus \scrT'$. Then 
\begin{gather}
\dim \spa (\scrT')\leq d_\mathrm{Prod}(U), \\
\dim \spa (\scrT'')\geq d-\dim \spa (\scrT')
\geq d-d_\mathrm{Prod}(U).
\end{gather}
The output state associated with each input state in $\scrT'$ is a product state, so  one measurement setting is required to verify it. 
By contrast, the output state associated with each input state in $\scrT''$ is entangled, 
so at least two measurement settings are required to verify it according to \thref{theorem:bipartite fewest settings}. Therefore,
\begin{align}
\mu(U)&\geq |\scrT'|+2|\scrT''|\geq \dim \spa (\scrT') +2 \dim \spa (\scrT'')\nonumber\\
&\geq \dim \spa (\scrT')+2 [d-\dim \spa (\scrT')]\nonumber\\
&=2d-\dim \spa (\scrT')\geq 2d-d_\mathrm{Prod}(U),
\end{align}
which implies \eref{eq:mu(U)proof}.

Next, suppose $d_\mathrm{Prod}(U)<d$. To prove \eref{eq:mu(U)}, it remains to prove the opposite inequality to \eref{eq:mu(U)proof}. 
Let $\scrS$ be a subset of  $\mathrm{Prod}(U)$ that is composed of $d_\mathrm{Prod}(U)$ linearly independent states.  By adding $d-d_\mathrm{Prod}(U)-1$ suitable product states, we can construct a set $\scrS'$ of $d-1$ linearly independent product states. 
Now we can add a product state that is not in the span of $\scrS'$ and is not orthogonal to any state in $\scrS'$. The resulting set $\scrS''$ forms a connected basis for $\caH$ and so can identify unitaries. 	In addition, the output state associated with each state in $\scrS$ is a product state and so can be verified by one measurement setting based on a local projective measurement. The output state associated with each state in $\scrS''\setminus\scrS$ can be verified by two measurement settings according to \thref{theorem:bipartite fewest settings}. Therefore,
\begin{align}
\mu(U)&\leq |\scrS|+2 |\scrS''\setminus\scrS|
=d_\mathrm{Prod}(U)+2[d-d_\mathrm{Prod}(U)]\nonumber\\
&=2d-d_\mathrm{Prod}(U),
\end{align}
which implies \eref{eq:mu(U)} given the opposite inequality in \eref{eq:mu(U)proof}. 

Now let us consider the case in which $d_\mathrm{Prod}(U)=d$. If the set $\mathrm{Prod}(U)$ is connected, then it contains a connected basis composed of product states
by \lref{lem:SpanBasis}. Moreover, the output state associated with each state in the basis is also a product state and so can be verified by one measurement setting. Therefore, $U$ can be verified by $d$ experimental settings, which means $\mu(U)=d$. 
 
If the set $\mathrm{Prod}(U)$ is not connected, then the set of test states of any valid verification protocol for $U$ contains at least one state not contained in $\mathrm{Prod}(U)$, which implies that $\mu(U)\geq d+1$ [cf. \thref{thm:minS-entFree} and the derivation that leads to \eref{eq:mu(U)proof}].  To complete the proof of \pref{pro:muUdprod}, it remains to construct a verification protocol for $U$ that requires only $d+1$ experimental settings. Let $\scrS$ be a subset of  $\mathrm{Prod}(U)$ that is composed of $d-1$ linearly independent states.  
We can add a product state that is not in the span of $\scrS$ and is not orthogonal to any state in $\scrS$. The resulting set $\scrS'$ forms a connected basis for $\caH$ and so can identify unitaries. 	In addition, the output state associated with each state in $\scrS$ is a product and so can be verified by one measurement setting based on a local projective measurement. The output state associated with the additional product state can be verified by two measurement settings according to \thref{theorem:bipartite fewest settings}. Therefore, $U$ can be verified by $d+1$ experimental settings, that is, $\mu(U)\leq d+1$. In conjunction with the opposite inequality derived above, we conclude that
 $\mu(U)= d+1$ when $d_\mathrm{Prod}(U)=d$ and the set $\mathrm{Prod}(U)$ is not connected.  
\end{proof}

\section{Proofs of \lsref{lem:SchmidtCoeffOrder}-\ref{lem:SchmidtCoeffSame} and and \crref{cor:SchmidtCoeffSame}}\label{app:lem:SchmidtCoeff}

\begin{proof}[Proof of \lref{lem:SchmidtCoeffOrder}]	
Let  $c_j=\cos\alpha_j$ and $s_j=\sin\alpha_j$	for $j=1,2,3$. Then the four Schmidt coefficients of the unitary operator $U(\alpha_1,\alpha_2,\alpha_3)$ can be expressed as follows:
	\begin{equation}\label{eq:SchmidtCoeffShort}
	\begin{aligned}
	|\zeta_0| = \sqrt{c_1^2 c_2^2 c_3^2 + s_1^2 s_2^2 s_3^2},\\
	|\zeta_1| = \sqrt{c_1^2 s_2^2 s_3^2 + s_1^2 c_2^2 c_3^2},\\
	|\zeta_2| = \sqrt{s_1^2 c_2^2 s_3^2 + c_1^2 s_2^2 c_3^2},\\
	|\zeta_3| = \sqrt{s_1^2 s_2^2 c_3^2 + c_1^2 c_2^2 s_3^2}.
	\end{aligned}
	\end{equation}
Now  the assumption $ 0 \leq \alpha_3 \leq \alpha_2 \leq \alpha_1 \leq \pi/4$ implies that
\begin{equation}
0\leq s_3\leq s_2\leq s_1
\leq \frac{\sqrt{2}}{2}\leq c_1\leq c_2\leq c_3\leq 1,
\end{equation}
which in turn implies that
\begin{equation}\label{eq:zi-zj}
\begin{aligned}
|\zeta_0|^2-|\zeta_1|^2=(c_1^2-s_1^2)(c_2^2 c_3^2-s_2^2 s_3^2) \ge 0,  \\
|\zeta_1|^2-|\zeta_2|^2=(c_3^2-s_3^2)( s_1^2 c_2^2-c_1^2 s_2^2) \ge 0,\\
|\zeta_2|^2-|\zeta_3|^2=(c_1^2-s_1^2)( s_2^2 c_3^2-c_2^2 s_3^2) \ge 0.
\end{aligned}
\end{equation}
Therefore,
\begin{equation}\label{eq:SchmidtCoeffOrderProof}
|\zeta_0| \ge |\zeta_1| \ge |\zeta_2| \ge |\zeta_3|\geq 0,
\end{equation}
which confirms \eref{eq:SchmidtCoeffOrder} in \lref{lem:SchmidtCoeffOrder}. The first inequality $|\zeta_0| \ge |\zeta_1|$  is saturated iff $c_1^2=s_1^2$ or $c_2^2 c_3^2=s_2^2 s_3^2$, which holds iff  $\alpha_1 = \pi/4$. 
The second inequality $|\zeta_1| \ge |\zeta_2|$ is saturated iff $c_3^2=s_3^2$ or $s_1^2c_2^2=c_1^2 s_2^2$, which holds iff $\alpha_2=\alpha_1$. 
The third inequality $|\zeta_2| \ge |\zeta_3|$  is saturated iff $c_1^2=s_1^2$ or $s_2^2c_3^2=c_2^2 s_3^2$, which holds iff $\alpha_1=\frac{\pi}{4}$ or $\alpha_3=\alpha_2$. Finally,
the last inequality $|\zeta_3|\geq 0$  is saturated iff
$s_1^2s_2^2=s_3^2=0$,  which holds iff 
$\alpha_2=\alpha_3=0$.
\end{proof}

\begin{proof}[Proof of \lref{lem:SchmidtCoeffSpecial}]
If  $0<\alpha_1=\alpha_2=\alpha_3 <\pi/4$, then \lref{lem:SchmidtCoeffOrder} implies that  $|\zeta_0|>|\zeta_1|=|\zeta_2|=|\zeta_3|>0$.

Next, suppose $|\zeta_0|>|\zeta_1|=|\zeta_2|=|\zeta_3|>0$. Then the inequality $|\zeta_0|>|\zeta_1|$ implies that $\alpha_1<\pi/4$ according to
\lref{lem:SchmidtCoeffOrder}; in addition, the equalities $|\zeta_1|=|\zeta_2|=|\zeta_3|$ imply that $\alpha_1=\alpha_2=\alpha_3$; finally, the inequality $|\zeta_3|>0$ implies that $\alpha_2>0$. 
Combining these results we can deduce that $0<\alpha_1=\alpha_2=\alpha_3 <\pi/4$, which completes the proof of \lref{lem:SchmidtCoeffSpecial}.
\end{proof}

\begin{proof}[Proof of \lref{lem:SchmidtCoeffSame}]
If  the condition in \eref{eq:SchmidtCoeffSameCon1}
holds, that is, $\alpha_j=\alpha_j'$ for $j=1,2,3$, then  $U(\alpha_1,\alpha_2,\alpha_3)$ and $U(\alpha_1',\alpha_2',\alpha_3')$ have the same Schmidt coefficients.  If the condition in \eref{eq:SchmidtCoeffSameCon2} holds, then 
$U(\alpha_1,\alpha_2,\alpha_3)$ and $U(\alpha_1',\alpha_2',\alpha_3')$ also have the same Schmidt coefficients according to \eref{eq:zetakSqSpecial}.

To prove the converse implication in \lref{lem:SchmidtCoeffSame},  let $C_j=\cos(2\alpha_j)$, $S_j=\sin(2\alpha_j)$, $C_j'=\cos(2\alpha_j')$, and $S_j'=\sin(2\alpha_j')$ for $j=1,2,3$;
then  the assumptions
 $ 0 \leq \alpha_3 \leq \alpha_2 \leq \alpha_1 \leq \pi/4$  and $0 \leq \alpha_3' \leq \alpha_2' \leq \alpha_1' \leq \pi/4$ imply  that
\begin{align}\label{eq:Corder}
0\leq C_1\leq C_2\leq C_3\leq 1, \quad 0\leq C_1'\leq C_2'\leq C_3'\leq 1.
\end{align}  
In addition, 
$C_j=0$  iff $\alpha_j=\pi/4$; similarly, $C_j'= 0$ iff $\alpha_j'=\pi/4$.  Furthermore,  according to \eref{eq:zetak}, the Schmidt coefficients of $U(\alpha_1,\alpha_2,\alpha_3)$ satisfy the following relations,
  \begin{align}\label{eq:dk_properties}
&|\zeta_0|^2+|\zeta_3|^2 = \half (1-C_1 C_2),\nonumber \\
&|\zeta_0|^2-|\zeta_2|^2 = \half C_2 (C_1+ C_3),\\ 
&|\zeta_0|^2-|\zeta_3|^2 = \half C_3 (C_1 + C_2),\nonumber
\end{align}
and the Schmidt coefficients of $U(\alpha_1',\alpha_2',\alpha_3')$ satisfy similar relations.

Suppose $U(\alpha_1,\alpha_2,\alpha_3)$ and $U(\alpha_1',\alpha_2',\alpha_3')$ have the same Schmidt coefficients. 
  Then \eref{eq:dk_properties} implies that
 \begin{equation}
	\begin{gathered}\label{eq:sameC}
	C_1 C_2=C_1' C_2', \\
	C_2 C_3=C_2' C_3',\\
	C_1 C_3=C_1' C_3'. 
	\end{gathered}
\end{equation} 
If $\alpha_3\leq \alpha_2\leq \alpha_1<\pi/4$, so that $C_3\geq C_2\geq C_1>0$, then \eref{eq:sameC} implies that $C_j'=C_j$ and  $\alpha_j'=\alpha_j$ for $j=1,2,3$, which confirms \eref{eq:SchmidtCoeffSameCon1}.

If $\alpha_1=\alpha_2=\pi/4$, then we have $C_1=C_2=0$, which implies that $C_1'=C_2'=0$ and $\alpha_1'=\alpha_2'=\alpha_1=\alpha_2=\pi/4$ given \esref{eq:Corder} and \eqref{eq:sameC}. In this case \eref{eq:SchmidtCoeffSameCon2} holds.
 
If $\alpha_1=\pi/4$ and $\alpha_3\leq \alpha_2<\pi/4$, then $C_1=0$ and $C_3\geq C_2>0$, which implies that $C_1'=0$, $C_3', C_2'>0$, and $\alpha_1'=\pi/4$ given \eref{eq:sameC}.  In addition, by virtue of \eref{eq:zetakSqSpecial} we can further 
deduce that $C_2 C_3=C_2' C_3'$ since $U(\alpha_1,\alpha_2,\alpha_3)$ and $U(\alpha_1',\alpha_2',\alpha_3')$ have the same Schmidt coefficients. So \eref{eq:SchmidtCoeffSameCon2} also holds in this case. 
\end{proof}

\begin{proof}[Proof of \crref{cor:SchmidtCoeffSame}]
	As shown in   \sref{sec:CanonicalForm}, $U$ is equivalent to $U(\alpha_1,\alpha_2,\alpha_3)$ or $U^*(\alpha_1,\alpha_2,\alpha_3)$ with the constraint $0 \leq \alpha_3 \leq \alpha_2 \leq \alpha_1 \leq \pi/4$, and $U'$ is equivalent to $U(\alpha_1',\alpha_2',\alpha_3')$ or $U^*(\alpha_1',\alpha_2',\alpha_3')$ with $0 \leq \alpha_3' \leq \alpha_2' \leq \alpha_1' \leq \pi/4$. By assumption $U(\alpha_1,\alpha_2,\alpha_3)$  and $U(\alpha_1',\alpha_2',\alpha_3')$ have the same Schmidt coefficients
	$s_0, s_1, s_2,  s_3$, which satisfy $s_0 > s_1 \ge s_2 \ge s_3$, so we  have $\alpha_1<\pi/4$ and $\alpha_1'<\pi/4$ by \lref{lem:SchmidtCoeffOrder}.
	In addition, $\alpha_j=\alpha_j'$ for $j=1,2,3$ and $U(\alpha_1,\alpha_2,\alpha_3)=U(\alpha_1',\alpha_2',\alpha_3')$ according to \lref{lem:SchmidtCoeffSame}. 
 Therefore,  $U'$ is equivalent to either $U$ or $U^*$ under local unitary transformations.
\end{proof}

\section{Two inequivalent unitary operators with the same Schmidt coefficients}\label{app:eg same coeff}

According to \eref{eq:zetakSqSpecial}, we can choose the following parameters
\begin{align}
\alpha_1&=\frac{\pi}{4},&
\alpha_2&=\arccos \sqrt{\frac{11}{16}},& \alpha_3&=\arccos \sqrt{\frac{17}{24}},\\
\alpha_1'&=\frac{\pi}{4},&
\alpha_2'&=\arccos\sqrt{\frac{5}{8}},& \alpha_3'&=\arccos \sqrt{\frac{13}{16}}, \end{align}
which satisfy $\cos(2\alpha_2')\cos(2\alpha_3')=\cos(2\alpha_2)\cos(2\alpha_3)$. 
It is easy to verify that the two inequivalent unitary operators $U(\alpha_1,\alpha_2,\alpha_3)$ and $U(\alpha_1',\alpha_2',\alpha_3')$
have the same Schmidt coefficients:
\begin{equation}
\sqrt{\frac{37}{128}},\quad \sqrt{\frac{37}{128}},\quad \sqrt{\frac{27}{128}},\quad \sqrt{\frac{27}{128}}. 
\end{equation}

\section{EFMIS for $U(\alpha_1,\alpha_2,\alpha_3)$ when  $0 \leq \alpha_3 \leq \alpha_2 \leq \alpha_1 \leq \pi/4$,  $\alpha_1>\alpha_3$, and  $\alpha_2<\pi/4$ }\label{app:case3-EFMIS}
In this appendix we construct an EFMIS for the unitary $U(\alpha_1,\alpha_2,\alpha_3)$ when  $0 \leq \alpha_3 \leq \alpha_2 \leq \alpha_1 \leq \pi/4$,  $\alpha_1>\alpha_3$, and  $\alpha_2<\pi/4$, which corresponds to the third case in the proof of \thref{thm:MinStwoqubit}.

Suppose in the magic basis the input state $|\phi_0\rangle$  has the expansion
$|\phi_0\rangle = \sum_{k=1}^4 \gamma_k |\Phi_k \rangle$, where the coefficients satisfy the normalization condition 
\begin{equation}
\sum_{k=1}^4 |\gamma_k|^2 =1. \label{eq:NormalCon}
\end{equation}
Then the product-state constraint holds if the coefficients $\gamma_1^2, \gamma_2^2, \gamma_3^2, \gamma_4^2$ can be expressed as follows:
  \begin{equation}\label{eq:case3-gamma}
  \begin{aligned}
  \gamma_1^2=&\rme^{2\rmi \alpha_1} \gamma_0^2,\quad 
  \gamma_2^2=\rme^{2\rmi \alpha_2} \gamma_0^2,\\
  \gamma_3^2=&\exp[\rmi (2\alpha_1+2\alpha_2-2\alpha_3+\pi)] \\
  &\times \frac{\sin(2\alpha_1+2\alpha_3)+\sin(2\alpha_2+2\alpha_3)}{\sin(2\alpha_1+2\alpha_2)} \gamma_0^2,\\
  \gamma_4^2=&\exp[\rmi (-2\alpha_3+\pi)] \\
  &\times \frac{\sin(2\alpha_1-2\alpha_3)+\sin(2\alpha_2-2\alpha_3)}{\sin(2\alpha_1+2\alpha_2)} \gamma_0^2,
  \end{aligned}
  \end{equation}
where 
\begin{align}\label{eq:case3-gamma0}
\gamma_0^2=&\frac{\sin(2\alpha_1+2\alpha_2)}{2\sin(2\alpha_1+2\alpha_2)+2[\sin(2\alpha_1)+\sin(2\alpha_2)]\cos(2\alpha_3)}
\end{align} 
is determined by the  normalization condition in \eref{eq:NormalCon}. Note that $\sin(2\alpha_1+2\alpha_2)>0$ by assumption.

Moreover, an EFMIS can be constructed as follows (in the magic basis):
\begin{equation}\label{eq:case3-inputApp}
\begin{aligned}
&|\phi_{1}\>=
\begin{pmatrix}
\gamma_{1} \\ \gamma_{2} \\ \gamma_{3} \\ \gamma_{4}
\end{pmatrix},\quad 
&|\phi_{2}\>=
\begin{pmatrix}
\gamma_{1} \\ -\gamma_{2} \\ \gamma_{3} \\ \gamma_{4}
\end{pmatrix}, \\
&|\phi_{3}\>=
\begin{pmatrix}
\gamma_{1} \\ \gamma_{2} \\ \gamma_{3} \\ -\gamma_{4}
\end{pmatrix},\quad 
&|\phi_{4}\>=
\begin{pmatrix}
-\gamma_{1} \\ \gamma_{2} \\ \gamma_{3} \\ -\gamma_{4}
\end{pmatrix}.
\end{aligned}
\end{equation}
The Gram matrix of the four states  reads 
  \begin{equation}\label{eq:gram matrix-case3}
  G = 
  \begin{pmatrix}
  1 & g_2 & g_4 & h_1 \\
  g_2 & 1 & h_2 & -g_3\\
  g_4 & h_2 & 1 & g_1 \\
  h_1 & -g_3 & g_1 & 1
  \end{pmatrix},
  \end{equation}
where $h_1 = g_1 + g_4-1, h_2=g_2+g_4-1 $ and $g_j=1-2|\gamma_j|^2$ for $j=1,2,3,4$. Its determinant is 
$64 |\gamma_1 \gamma_2 \gamma_3 \gamma_4|^2 \ne 0$, which implies that the four states in \eref{eq:case3-inputApp} span the whole Hilbert space.  
In addition, we have 
  \begin{align}\label{eq:gne0}
  g_1,g_2,g_3,g_4\ne0
  \end{align}
as proved below, which means  the corresponding transition graph is connected, so the states in \eref{eq:case3-inputApp} indeed form an EFMIS.

\begin{proof}[Proof of \eref{eq:gne0}]
We shall prove \eref{eq:gne0} by reduction to absurdity. Suppose $g_1=0$ or $g_2=0$; then we have $|\gamma_0|^2=1/2$. Let $S_j=\sin(2\alpha_j)$ and $C_j=\cos(2\alpha_j)$ for $j=1,2,3$. From \eref{eq:case3-gamma0}, we can deduce that
\begin{equation}
  (S_1+S_2)C_3=0.
\end{equation}
Therefore,  $\alpha_1=\alpha_2=\alpha_3=0$ or $\pi/4$, which contradicts the assumption. This contradiction shows that $g_1 \ne 0$ and $g_2 \ne 0$. 

Suppose $g_3=0$; then  $|\gamma_3|^2=1/2$. From \esref{eq:case3-gamma} and \eqref{eq:case3-gamma0} we can deduce that
\begin{equation}\label{eq:g3=0}
  C_1S_3+C_2S_3=C_1S_2+S_1C_2.
\end{equation}
Meanwhile, the assumptions  $\alpha_1>\alpha_3$ and $\alpha_2<\pi/4$ imply that  $C_2>0$, $S_2\ge S_3$, $S_1>S_3$, and
\begin{equation}
  C_1S_3+C_2S_3<C_1S_2+S_1C_2,
\end{equation}
which contradicts \eref{eq:g3=0}. This contradiction shows that $g_3 \ne 0$.

Suppose $g_4=0$; then $|\gamma_4|^2=1/2$. From \esref{eq:case3-gamma} and \eqref{eq:case3-gamma0} we can deduce that
\begin{equation}\label{eq:g4=0}
  -(C_1+C_2)S_3=\sin(2\alpha_1+2\alpha_2).
\end{equation}
However, this equation cannot hold given the assumptions $\alpha_1>\alpha_3$ and $\alpha_2<\pi/4$. 
This contradiction shows that $g_4 \ne 0$ and completes the proof of \eref{eq:gne0}. 
\end{proof}

\bigskip

\section{Proof of \crref{cor:MinStwoqubit-general}}\label{app:cor:MinStwoqubit-general}

\begin{proof}
\Crref{cor:MinStwoqubit-general} follows from \thref{thm:MinStwoqubit} and the following equations:
\begin{align}
\mu(U(\alpha_1+\pi/2,\alpha_2,\alpha_3))&=\mu(U (\alpha_1,\alpha_2,\alpha_3)), \label{eq:period1}\\
\mu(U(\alpha_1,\alpha_2+\pi/2,\alpha_3))&=\mu(U (\alpha_1,\alpha_2,\alpha_3)),\label{eq:period2}\\
\mu(U(\alpha_1,\alpha_2,\alpha_3+\pi/2))&=\mu(U (\alpha_1,\alpha_2,\alpha_3)),\label{eq:period3}\\
\mu(U (\pi/4-\alpha_1,\alpha_2,\alpha_3))&=\mu(U(\pi/4+\alpha_1,\alpha_2,\alpha_3)),\label{eq:reflection1}\\
\mu(U (\alpha_1,\pi/4-\alpha_2,\alpha_3))&=\mu(U(\alpha_1,\pi/4+\alpha_2,\alpha_3)),\label{eq:reflection2}\\
\mu(U (\alpha_1,\alpha_2,\pi/4-\alpha_3))&=\mu(U(\alpha_1,\alpha_2,\pi/4+\alpha_3)).\label{eq:reflection3}
\end{align}
\Esref{eq:period1}-\eqref{eq:period3} mean $\mu(U (\alpha_1,\alpha_2,\alpha_3))$ is periodic in $\alpha_1,\alpha_2,\alpha_3$, respectively, with the common period of $\pi/2$. \Esref{eq:reflection1}-\eqref{eq:reflection3} mean
$\mu(U (\alpha_1,\alpha_2,\alpha_3))$ is invariant under reflection with respect to the three planes specified by $\alpha_1=\pi/4$, $\alpha_2=\pi/4$, $\alpha_3=\pi/4$, respectively.

\Eref{eq:period1} follows from the equality
\begin{equation}
  U (\alpha_1,\alpha_2,\alpha_3)=\rmi (\sigma_1 \otimes \sigma_1)U(\alpha_1+\pi/2,\alpha_2,\alpha_3),
\end{equation}
given that  $\mu(U)$ is invariant under local unitary transformations. \Esref{eq:period2} and \eqref{eq:period3} can be proved in a similar way.

\Eref{eq:reflection1} follows from the equality
\begin{equation}
  U (\pi/4-\alpha_1,\alpha_2,\alpha_3) = -\rmi \sigma_1^\rmA  U^*(\pi/4+\alpha_1,\alpha_2,\alpha_3) \sigma_1^\rmB,
\end{equation}
given that  $\mu(U)$ is also invariant under complex conjugation. \Esref{eq:reflection2} and \eqref{eq:reflection3} can be proved in a similar way. 
\end{proof}

\bibliography{MinS_ref}
\end{document}